\documentclass[letterpaper]{article} % DO NOT CHANGE THIS
\usepackage{aaai20}  % DO NOT CHANGE THIS
\usepackage{times}  % DO NOT CHANGE THIS
\usepackage{helvet} % DO NOT CHANGE THIS
\usepackage{courier}  % DO NOT CHANGE THIS
\usepackage[hyphens]{url}  % DO NOT CHANGE THIS
\usepackage{graphicx} % DO NOT CHANGE THIS
\usepackage{graphicx}  % DO NOT CHANGE THIS
\nocopyright
\title{Refining Tournament Solutions via Margin of Victory}

\author{Markus Brill \\
Efficient Algorithms \\
TU Berlin 
\And
Ulrike Schmidt-Kraepelin \\
Efficient Algorithms \\
TU Berlin 
\And
Warut Suksompong \\
Department of Computer Science \\
University of Oxford
}

\usepackage{url}
\usepackage[utf8]{inputenc}
\usepackage{subfig}
\usepackage{dblfloatfix}
\usepackage{amsmath}
\usepackage{amsthm}
\usepackage{booktabs}
\usepackage{xspace}
\usepackage{cleveref}
\usepackage{tabularx}
\usepackage{amssymb}
\usepackage{thmtools}
\usepackage{thm-restate}

\usepackage[cmyk,x11names]{xcolor}
\usepackage{tikz}
\usetikzlibrary{arrows.meta}
\tikzset{>={Latex[width=2mm,length=2mm]}}
\usetikzlibrary{decorations.pathmorphing,shapes,snakes}
\usetikzlibrary{calc}
\usetikzlibrary{arrows, decorations.markings}

\newtheorem{theorem}{Theorem}
\newtheorem{corollary}[theorem]{Corollary}

\theoremstyle{definition}
\newtheorem{definition}[theorem]{Definition}
\newtheorem{example}[theorem]{Example}

\newtheorem*{claim}{Claim}

\newenvironment{proofsketch}{%
  \proof}{\endproof}

\newcommand{\cp}{\ensuremath{\mathit{CO}}\xspace}
\newcommand{\tc}{\ensuremath{\mathit{TC}}\xspace}
\newcommand{\uc}{\ensuremath{\mathit{UC}}\xspace}
\newcommand{\ba}{\ensuremath{\mathit{BA}}\xspace}

\DeclareMathOperator*{\argmax}{arg\,max}

\usepackage[textsize=tiny,textwidth=1.5cm]{todonotes}

\newcommand{\ie}{i.e.,\xspace}
\newcommand{\eg}{e.g.,\xspace}
\newcommand{\midd}{\mathrel{:}}

\newcommand{\citet}[1]{\citeauthor{#1}~(\citeyear{#1})} 
   
\newcommand{\citep}[1]{\cite{#1}}					   

\newcounter{int}
\makeatletter
\newcommand{\citen}[1] {\setcounter{int}{0}\@for\tmp:=#1\do{%
\ifnum \value{int}>0; \fi%
\setcounter{int}{1}%
\citeauthor{\tmp} \shortcite{\tmp}}}
\makeatother

\makeatletter
\newcommand{\citenp}[1]{\setcounter{int}{0}\@for\tmp:=#1\do{%
\ifnum \value{int}>0; \fi%
\setcounter{int}{1}%
\citeauthor{\tmp}, \citeyear{\tmp}}}
\makeatother

\newcommand{\repeatcaption}[2]{%
  \renewcommand{\thefigure}{\ref{#1}}%
  \captionsetup{list=no}%
  \caption{#2 (repeated from page \pageref{#1})}%
}

\begin{document}

\maketitle

\begin{abstract}
Tournament solutions are frequently used to select winners from a set of alternatives based on pairwise comparisons between alternatives.
Prior work has shown that several common tournament solutions tend to select large winner sets and therefore have low discriminative power.
In this paper, we propose a general framework for refining tournament solutions.
In order to distinguish between winning alternatives, and also between non-winning ones, we introduce the notion of \emph{margin of victory (MoV)} for tournament solutions. 
MoV is a robustness measure for individual alternatives: For winners, the MoV captures the distance from dropping out of the winner set, and for non-winners, the distance from entering the set. In each case, distance is measured in terms of which pairwise comparisons would have to be reversed in order to achieve the desired outcome.
For common tournament solutions, including the top cycle, the uncovered set, and the Banks set, we determine the complexity of computing the MoV and provide worst-case bounds on the MoV for both winners and non-winners.
Our results can also be viewed from the perspective of bribery and manipulation.
\end{abstract}

\section{Introduction}

A number of practical choice scenarios involving pairwise comparisons of alternatives can be modeled using tournaments. For instance, the pairwise comparisons could represent match outcomes when alternatives are teams in a round-robin sports competition, or the results of pairwise majority comparisons when the alternatives are candidates in an election. 
In order to select the set of ``winners'' from a tournament, several methods, known in the literature as \emph{tournament solutions}, have been proposed. Over the past decades, many of these tournament solutions have been extensively studied from both the axiomatic and the computational point of view \citep{Laslier97,BrandtBrHa16}. Due to their generality and wide range of applications, the study of tournament solutions and their properties has attracted considerable attention 
in recent years 
(e.g., \citenp{BBH16a,ABF+15a,MSY15a,Dey17a}).

Although tournament solutions provide a rich supply of procedures for choosing tournament winners according to various criteria, they often exhibit low discriminative power because the chosen winner sets tend to be large.
Indeed, previous work has shown that common tournament solutions such as the top cycle, the uncovered set, the Banks set, and the minimal covering set almost never exclude any alternative in a random tournament \citep{Fey08,ScottFe12}, while the bipartisan set includes on average half of the alternatives in the winner set \citep{FisherRy95}.\footnote{These results assume that tournaments are chosen from the uniform distribution. \citet{BrSe15a} and \citet{SaileSu18} relax this assumption and study the discriminative power of tournament solutions when tournaments are generated according to different stochastic models.}$^{,}$\footnote{\citet{BBSS14a} show that any tournament solution satisfying the property of \emph{stability}, including the top cycle and the bipartisan set, chooses at least half of the alternatives on average.}
This naturally raises the question of how tournament solutions can be refined in order to differentiate among the winners of a given tournament.

In this paper, we propose a general framework for refining tournament solutions and for distinguishing among the winners---as well as among the non-winners---of a tournament.
We introduce the concept of \emph{margin of victory (MoV)} for tournaments, which captures how close a winner is to dropping out of the winner set, and by symmetry how close a non-winner is to entering the winner set.
For a given tournament and weights on the tournament edges, the MoV of a winner is defined as the minimum total weight of edges whose reversals take it out of the winner set.
Analogously, the MoV of a non-winner is defined as the negative of the minimum total weight of edges whose reversals bring it into the winner set.
An important special case is when the edges are unweighted: in this case, the problem reduces to finding the minimum number of edges to be reversed.

The edge weights in our MoV framework can be interpreted in a number of different ways.
Generally speaking, they represent the strength of the edges or the cost that one incurs by reversing them.
In an election, a weight may reflect the proportion of voters who agree with the corresponding pairwise comparison, while in a sports competition, it may indicate the gap between the two teams in the match result.
Alternatively, our refinements can also be viewed through the lens of bribery and manipulation.
In this context, the weights express the amount of bribe that a manipulator needs to pay in order to reverse a pairwise comparison; the recipients of the bribe are voters in the case of an election and teams or referees in the case of a sports competition.

\begin{table*}
\begin{center}
\begin{tabular}{ l c c c c c c c c c }
\toprule
 & \multicolumn{2}{ c }{MoV for Winners} && \multicolumn{2}{ c }{MoV for Non-Winners} && \multicolumn{2}{ c }{Bounds on MoV}  &\\
\cmidrule{2-3} \cmidrule{5-6} \cmidrule{8-9}
& unweighted & weighted && unweighted & weighted && lower bound & upper bound&\\
% \cmidrule{2-9}
\midrule
Copeland (\cp) & P (\ref{thm:destructive-CP-P}) & P (\ref{thm:destructive-CP-P}) && P (\ref{thm:constructive-CL}) & P (\ref{thm:constructive-CL}) && $-(n-2)$ (\ref{thm:constructive-bound-co}) & $\lfloor n/2 \rfloor$ (\ref{thm:destructive-bounds}) &\\
Top Cycle (\tc) & P (\ref{cor:MoV-UC-kings}) & P (\ref{cor:MoV-UC-kings}) && P (\ref{thm:constructive-weighted-tc}) & P (\ref{thm:constructive-weighted-tc}) && $-1$ (\ref{thm:constructive-bound-tc}) & $\lfloor n/2 \rfloor$ (\ref{thm:destructive-bounds}) &\\
Uncovered Set (\uc) & P (\ref{cor:MoV-UC-kings}) & P (\ref{cor:MoV-UC-kings}) && $n^{O(\log n)}$ (\ref{thm:constructive-UC-unweighted}) & NP-h (\ref{thm:constructive-UC-NP}) && $-\lceil \log_2 n\rceil$ (\ref{thm:constructive-bound-uc-ba}) & $\lfloor n/2\rfloor$ (\ref{thm:destructive-bounds}) &\\
$3$-kings & P (\ref{cor:MoV-UC-kings}) & P (\ref{cor:MoV-UC-kings}) && P (\ref{thm:constructive-bound-tc}) & NP-h (\ref{thm:constructive-weighted-kkings-NP}) && $-1$ (\ref{thm:constructive-bound-tc}) & $\lfloor n/2 \rfloor$ (\ref{thm:destructive-bounds}) &\\
$k$-kings (for $k \ge 4$) & NP-h (\ref{thm:destr:kkings:nphard}) & NP-h (\ref{thm:destr:kkings:nphard}) && P (\ref{thm:constructive-bound-tc}) & NP-h (\ref{thm:constructive-weighted-kkings-NP}) && $-1$ (\ref{thm:constructive-bound-tc}) & $\lfloor n/2 \rfloor$ (\ref{thm:destructive-bounds}) &\\
Banks set (\ba) & NP-h (\ref{thm:destructive-banks-NP}) & NP-h (\ref{thm:destructive-banks-NP}) && NP-h (\ref{thm:constructive-banks-NP}) & NP-h (\ref{thm:constructive-banks-NP}) && $-\lceil \log_2 n\rceil$ (\ref{thm:constructive-bound-uc-ba}) & $\lfloor n/2 \rfloor$ (\ref{thm:destructive-bounds}) &\\
\bottomrule
\end{tabular}
\caption{Overview of our results, with $n$ denoting the number of alternatives in the tournament. The computational results for Copeland (first four entries) also follow from \citet{FHHR09c}; for completeness, we give proofs tailored to our setting. The numbers in parentheses refer to the corresponding theorem or corollary numbers.}
\label{tab:results}
\end{center}
\end{table*}

\subsection{Our Results}

We study the computational complexity of the MoV with respect to several common tournament solutions, including the Copeland set, the top cycle, the uncovered set, and the Banks set. 
For each tournament solution, we determine the complexity of computing the MoV for both winners and non-winners, in both the unweighted and weighted setting.
In addition, we derive tight or asymptotically tight lower and upper bounds on the MoV for all of the considered tournament solutions in the unweighted setting. 
An overview of our results can be found in Table~\ref{tab:results}.

\subsection{Related Work} \label{subsec:relatedwork}

\citet{KrAi17b} considered refinements of several tournament solutions based on their \emph{binary tree representations}. This approach can only be applied to tournament solutions that admit such a representation, and different representations may yield different refinements.

While our work is the first to consider a MoV concept for tournament solutions (to the best of our knowledge), a related notion with the same name has been extensively explored in the context of voting.
Unlike in our setting where the MoV serves the purpose of distinguishing among alternatives, in voting the MoV is typically used to measure the robustness of election outcomes \citep{Cary11,MagrinoRiSh11,Xia12x,DeNa15a}. 
As such, the notion there is defined for election outcomes as a whole rather than for individual alternatives. The same holds for the robustness measure of \citet{SYE13a}.

A long line of work has investigated various forms of bribery and manipulation in tournaments.
This includes manipulating the tournament bracket to help a certain candidate win the tournament \citep{VuAlSh09,Vassilevskawilliams10,KSV17a,AzizGaMa18} and bribing players to intentionally lose matches \citep{RussellWa09,MatteiGoKl15,Kiwi15a,KonickiVa19}.
In particular, \citet{RussellWa09} considered a model where only a given subset of edges can be reversed, while other edges are assumed to be fixed. This constitutes a special case of our weighted setting, with sufficiently high weights on fixed edges.

In the context of bribery in voting,
\citet{FHHR09c} considered a ``microbribery'' setting in which voters can be bribed to change individual pairwise comparisons between candidates, even if this results in intransitive preferences of the voter. This corresponds to our weighted setting, with weights given by pairwise majority margins.

Finally, a closely related problem is that of finding possible (resp., necessary) winners of partially specified tournaments: Given a tournament with some missing edges, the goal is to determine whether a certain alternative can be a winner for some (resp., all) completions of the tournament \citep{ABF+15a}.
We observe that both problems can be reduced to computing the MoV in the weighted setting, by considering an arbitrary completion of the partial tournament and making the original edges prohibitively expensive to reverse.

\section{Preliminaries}
\label{sec:prelim}

A \emph{tournament} $T=(V,E)$ is a directed graph such that there is exactly one directed edge between every pair of vertices. The vertices of a tournament $T$, denoted $V(T)$, are often referred to as \textit{alternatives}. Let $n=|V(T)|$. The set of directed edges of $T$, denoted $E(T)$, represents an asymmetric and connex \emph{dominance relation} on the set of alternatives. An alternative $x$ is said to \emph{dominate} another alternative $y$ if $(x,y)\in E(T)$ (i.e., there is a directed edge from~$x$ to~$y$). When the tournament is clear from the context, we often write $x \succ y$ to denote $(x,y) \in E(T)$. By definition, for each pair $x,y$ of distinct alternatives, either $x$ dominates $y$ ($x \succ y$) or $y$ dominates $x$ ($y \succ x$), but not both. 

For a given tournament $T$ and $x \in V(T)$,
the \emph{dominion} of $x$, denoted by $D(x)$, is defined as the set of alternatives $y$ such that $x\succ y$.
Similarly, the set of \emph{dominators} of $x$, denoted by $\overline{D}(x)$, is defined as the set of alternatives $y$ such that $y\succ x$.
An alternative $x \in V(T)$ is said to be a \emph{Condorcet winner} in $T$ if it dominates every other alternative (i.e., $D(x) = V(T)\setminus \{x\}$), and a \emph{Condorcet loser} in $T$ if it is dominated by every other alternative. % (i.e., $\overline{D}(x) =  V(T)\setminus \{x\}$). 
See \Cref{fig:example} for an example tournament. 

The dominance relation can be extended to sets by writing $X\succ Y$ if $x\succ y$ for all $x\in X$ and $y\in Y$.
A set $X \subseteq V(T)$ is called a \emph{dominating set} in $T$ if
every alternative outside of $X$ is dominated by at least one alternative in $X$.

For $U \subseteq V(T)$, $T|_U$ denotes the restriction of $T$ to $U$, and $T_{-x}$ is short for $T|_{V(T)\setminus \{x\}}$.
For an edge $e=(x,y)$, we let $\overline{e}$ denote its reversal, i.e., $\overline{e} = (y,x)$. Similarly, for a set of edges $R \subseteq E(T)$, we define 
$\overline{R} = \{\overline{e}: e \in R\}$.

A \emph{tournament solution} is a function that maps each tournament to a nonempty subset of its alternatives, usually called the set of \emph{winners} or the \emph{choice set}.
The set of winners of a tournament $T$ with respect to a tournament solution $S$ is denoted by $S(T)$.
The tournament solutions considered in this paper are as follows:

\begin{itemize}
\item The \emph{Copeland set} (\cp) is the set of alternatives with the largest dominion, i.e., $\cp(T) = \argmax_{x \in V(T)}|D(x)|$. 
\item The \emph{top cycle} (\tc) is the (unique) smallest set $B$ of alternatives such that $B \succ V(T) \setminus B$.
Equivalently, the top cycle is the set of alternatives that can reach every other alternative via a directed path.
\item The \emph{uncovered set} (\uc), is the set of alternatives that are not ``covered'' by any other alternative. 
An alternative $x$ \emph{covers} another alternative $y$ if $D(y) \subseteq D(x)$. 
Equivalently, \uc is the set of alternatives that can reach every other alternative via a directed path of length at most two.
\item The set of \emph{$k$-kings}, for an integer $k\geq 3$, is the set of alternatives that can reach every other alternative via a directed path of length at most $k$.
\item The \emph{Banks set} (\ba) is the set of alternatives that appear as the maximal element of some inclusion-maximal transitive subtournament.\footnote{A transitive subtournament is \emph{inclusion-maximal} if it is not contained in any other transitive subtournament.
If an alternative $x$ dominates all alternatives in a transitive subtournament $T'$, we say that $x$ \emph{extends} $T'$.}
\end{itemize}

All of these tournament solutions satisfy \emph{Condorcet-consistency}: Whenever a Concorcet winner exists, it is chosen as the unique winner. 

It is clear from the definitions that \uc (the set of ``$2$-kings'') is contained in the set of $k$-kings for any $k\geq 3$, which is in turn a subset of \tc.
Moreover, both \cp and \ba are contained in \uc \citep{BrandtBrHa16}. 

\def\nd{1.5em} 	
\def\ra{0.75em} 

\begin{figure}
\centering
\vspace{-0.1cm}
\begin{tikzpicture}[node distance=\nd,vertex/.style={circle,draw,minimum size=2*\ra,inner sep=0pt}]
		\node[vertex] (a) 	{$a$};
		\node[vertex] (b) [right=of a] {$b$};
		\node[vertex] (c) [right=of b] {$c$};
		\node[vertex] (d) [right=of c] {$d$};
		\node[vertex] (e) [right=of d] {$e$};
		\node[vertex] (f) [right=of e] {$f$};
		\draw [-latex] (c) to[bend left=45] (f);
		\draw [-latex] (b) to[bend right=45] (e);
	\end{tikzpicture}	
	\caption{Tournament $T$ with $V(T)=\{a,b,c,d,e,f\}$. All omitted edges are assumed to point from right to left (\eg $D(f)=\{a,b,d,e\}$ and $a$ is a Condorcet loser in~$T$).} 
	\label{fig:example}
\end{figure}
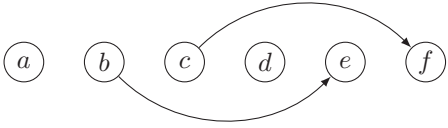

\section{Margin of Victory in Tournaments}
\label{sec:mov}

\newcommand{\mov}{\ensuremath{\mathsf{MoV}}}
\newcommand{\crs}{CRS\xspace}
\newcommand{\drs}{DRS\xspace}

We define the margin of victory (MoV) for a winning (resp., non-winning) alternatives in terms of sets of edges whose reversals result in the alternative becoming a non-winner (resp., winner). Edge sets with this property will be called destructive (resp., constructive) reversal sets. %%FN 
To formally define these notions, we need some notation. 
For a tournament $T$ and a set $R \subseteq E(T)$ of edges, we let $T^R$ denote the tournament that results from $T$ when reversing all edges in~$R$, i.e., 
$V(T^R) = V(T)$ and $E(T^R) = (E(T) \setminus R) \cup \overline{R}$.

Fix a tournament solution $S$ and consider a tournament~$T$. An edge set $R \subseteq E(T)$ is called a \textit{destructive reversal set (\drs)} for $x \in S(T)$ if $x \notin S(T^R)$. Analogously, $R$ is called a \textit{constructive reversal set (\crs)} for $x \in V(T)\setminus S(T)$ if $x \in S(T^R)$.
\footnote{The terms ``destructive'' and ``constructive'' are borrowed from the literature on control and bribery in voting (\eg \citenp{FaRo15a}), where the goal is either to prevent a given candidate from winning (destructive control/bribery) or to make a given candidate a winner (constructive control/bribery).} 

In general, destructive and constructive reversal sets are not unique, and finding \emph{some} \drs or \crs is usually easy. For example, for all Condorcet-consistent tournament solutions $S$, a straightforward CRS for an alternative $x \notin S(T)$ is given by  $R=\{(y,x) \midd y \in \overline{D}(x)\}$. This is because $x$ is a Condorcet winner in $T^R$.

We furthermore assume that we are given a weight function $w:E(T) \rightarrow \mathbb{R}^+$ that assigns a positive weight $w(e)> 0$ to each edge $e \in E(T)$.\footnote{We forbid zero-weight edges for technical reasons. Their existence can be imitated by setting their cost to a small $\epsilon>0$.} 
The weight of an edge can be thought of as the \emph{cost} that is incurred by reversing the edge. 
The cost of a set $R \subseteq E(T)$ is $w(R) = \sum_{e \in R} w(e)$. 

A natural special case is the setting in which reversing is equally costly for all edges. In this \emph{unweighted} setting, we assume $w(e)=1$ for all $e \in E(T)$, and finding a minimum cost reversal set reduces to finding a reversal set of minimum cardinality.

We are now ready to define the main concept of this paper.

\begin{definition}
For a tournament solution $S$ and a tournament~$T$, the \emph{margin of victory} of  $x \in S(T)$ is given by
\[ 
\mov_S(x,T) = \min \{w(R) \midd R \text{ is a \drs for $x$ in $T$} \} \text,
\]
and for an alternative $x \in V(T) \setminus S(T)$, it is given by
\[ 
\mov_S(x,T) = -\min \{w(R): R \text{ is a \crs for $x$ in $T$} \} \text. 
\]
\end{definition}

By definition, $\mov_S(x,T)$ is positive if $x \in S(T)$, and negative otherwise. In the unweighted setting, all MoV values are (positive or negative) integers.

\begin{table}[t]
\begin{center}
\begin{tabular}{r c c c c c c }
\toprule
                  & $a$ & $b$ & $c$ & $d$ & $e$ & $f$  \\ 
\midrule
$\mov_\uc(x,T)$   & $-2$ & $-1$ & $1$ & $1$ & $1$ & $2$  \\ 
min DRS/CRS       & $fa,da$ & $fb$ & $cf$ & $dc$ & $ed$ & $fe,fb$  \\ 
\bottomrule
\end{tabular}
\caption{MoV values and minimal reversal sets with respect to $\uc$ for the tournament $T$ in \Cref{fig:example} (unweighted setting). For improved readability of reversal sets, we omit set brackets and use $xy$ to denote edge $(x,y)$.
}
\label{tab:example}
\end{center}
\end{table}

\begin{example}\label{example}
Consider the tournament $T$ in \Cref{fig:example}. It can be easily verified that $\uc(T)=\{c,d,e,f\}$. For the unweighted setting, \Cref{tab:example} gives the MoV values for this tournament with respect to the uncovered set, together with examples of minimal destructive or constructive reversal sets.  
\end{example}

Note that minimal reversal sets are generally not unique, and that a minimal reversal set for an alternative $x$ may exclusively consist of edges \emph{not} incident to $x$ (\eg $\{(f,e)\}$ is a minimal CRS for $b$ in Example \ref{example}).

\section{Computing the MoV for Winners}\label{sec:compWinners}

We now study the computational complexity of computing the MoV for winners. 
We are given a tournament $T$, a weight function $w: E(T) \rightarrow \mathbb{R}^+$, a tournament solution $S$, and an alternative $x \in S(T)$; the task is to compute $\mov_S(x,T)$. 
Clearly, a polynomial-time algorithm for the weighted setting also applies to the unweighted setting, while a hardness result in the unweighted setting implies one for the weighted setting.
We note that in all cases where we provide a polynomial-time algorithm (\ie table entries ``P'' in \Cref{tab:results}), our algorithm not only determines the MoV value, but also finds a minimum DRS (or CRS when considering non-winners). 
Omitted proofs can be found in the appendix.

\subsection{Copeland}

The MoV for Copeland has already been studied (under different names) in slightly different settings \citep{FHHR09c,RussellWa09}. 
In particular, Theorem~3.7 of \citet{FHHR09c} implies that the MoV for Copeland winners can be computed efficiently whenever the weights correspond to pairwise majority margins resulting from a preference profile. For completeness, we provide a (simpler) proof tailored to our setting.\footnote{For the \textit{unweighted} case, a greedy approach suffices to compute the MoV of a Copeland winner. This case is not particularly interesting, however, as it can be easily verified that $\mov_\cp(x,T)=1$ for all $x \in \cp(T)$ whenever $|\cp(T)|>1$.}

\begin{restatable}{theorem}{restateCopelandDestructive}
\label{thm:destructive-CP-P}
Computing the MoV of a \cp winner in the weighted setting can be done in polynomial time.
\end{restatable}

\subsection{Uncovered Set, $\boldsymbol{k}$-Kings and Top Cycle}

\begin{figure*}[ht!]
\centering
\scalebox{1}{
\begin{tikzpicture}
\node[circle,fill,inner sep=2pt](s) at (0,0) {};\node[below=2pt] at (s){$x$};
\node[circle,fill,inner sep=2pt](v1) at (1,0) {}; 
\node[circle,fill,inner sep=2pt](v2) at (2,0) {};
\node[circle,fill,inner sep=2pt](v3) at (3,0) {}; 
\node[circle,fill,inner sep=2pt](v4) at (4,0) {};
\node[circle,fill,inner sep=2pt](v5) at (5,0) {}; 
\node[circle,fill,inner sep=2pt](v6) at (6,0) {}; 
\node[circle,fill,inner sep=2pt](v7) at (7,0) {}; 
\node[circle,fill,inner sep=2pt](v8) at (8,0) {}; 
\node[circle,fill,inner sep=.5pt](do1) at (8.4,0) {};
\node[circle,fill,inner sep=.5pt](do2) at (8.5,0) {};
\node[circle,fill,inner sep=.5pt](do3) at (8.6,0) {};
\node[circle,fill,inner sep=2pt](v9) at (9,0) {}; 
\node[circle,fill,inner sep=2pt](v10) at (10,0) {}; 

%SINK
\node[circle,fill,inner sep=2pt](t) at (11,0) {};\node[below=2pt] at (t){$y$}; 

%EDGES 
\draw[->,thick,dashed] (v1) -- (v2) node [near start,above] {$\overline{e}_1$};
\draw[->,thick,dashed] (v4) -- (v5)  node [near start,above] {$\overline{e}_2$};
\draw[->,thick,dashed] (v6) -- (v7) node [near start,above] {$\overline{e}_3$};
\draw[->,thick,dashed] (v9) -- (v10)  node [near start,above] {$\overline{e}_\ell$};

\draw[->,thick] (s) -- (v1) node [midway,below] {$P_0$};
\draw[->,thick] (v2) -- (v3);
\draw[->,thick] (v3) -- (v4);
\draw[->,thick] (v5) -- (v6);
\draw[->,thick] (v7) -- (v8);
\draw[->,thick] (v10) -- (t) node [near start,above] {$P_\ell$};

%Paths

\draw[black!60, thick, -> , decorate, decoration = {snake, amplitude=.3mm, post length=1mm}] (s)  to [bend left=50] (v2);
\draw[black!60, thick, -> , decorate, decoration = {snake, amplitude=.3mm, post length=1mm}] (s)  to [bend left=40](v5);
\draw[black!60, thick, -> , decorate, decoration = {snake, amplitude=.3mm, post length=1mm}] (s)  to [bend left=40](v7);
\draw[black!60, thick, -> , decorate, decoration = {snake, amplitude=.3mm, post length=1mm}] (s)  to [bend left=40] node[near end, right=20pt] {\footnotesize$Q_\ell^{(1)}$}(v10);

\draw[black!60, thick, -> , decorate, decoration = {snake, amplitude=.3mm, post length=1mm}] (v1)  to [bend right=30]node[near start,left=20pt] {\footnotesize$Q_1^{(2)}$}(t);
\draw[black!60, thick, -> , decorate, decoration = {snake, amplitude=.3mm, post length=1mm}] (v4)  to [bend right=30](t);
\draw[black!60, thick, -> , decorate, decoration = {snake, amplitude=.3mm, post length=1mm}] (v6)  to [bend right=30](t);
\draw[black!60, thick, -> , decorate, decoration = {snake, amplitude=.3mm, post length=1mm}] (v9)  to [bend right=30](t);

\draw [decorate,decoration={brace,amplitude=10pt,mirror, raise=4pt},yshift=0pt, thick] (v2) -- (v4) node [black,midway,yshift=-0.6cm, xshift=.4cm] {\footnotesize $P_1$};

\end{tikzpicture}}
\caption{An illustration of a situation in which Lemma \ref{lem:xcut-kkings-destr} would be violated. The path $P$ which occurs after the reversion of edges in $R$ is illustrated by straight arcs, where edges that were reversed are dashed. Paths $Q_i$ with the property that they contain exactly one reversed edge, namely $e_i$, are depicted by bended, curled arrows.}
\label{fig:destructive-kkings}
\end{figure*}
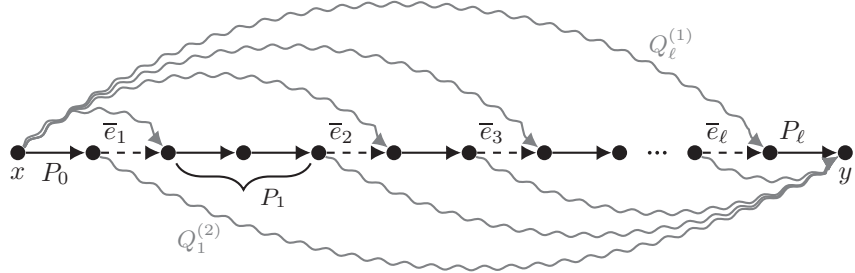

The problems of computing the MoV for \uc, $k$-kings and \tc are not only closely related to each other but also to theory on network flows. Since \uc can be interpreted as $2$-kings and \tc as $(n-1)$-kings, we only refer to $k$-kings and assume that $k$ can be chosen from $\{2, \dots, n-1\}$. A \drs for $x$ is then an edge set $R$ such that $x$ has distance greater than $k$ to at least one alternative $y$ in $T^R$.

Finding a minimum \drs is closely related to finding $\ell$-length bounded $s$-$t$-cuts of minimum capacity. In the latter problem, we are given a directed network $G=(V,E)$ with a capacity function $u:E \rightarrow \mathbb{R}^{+}$, two distinguished nodes $s,t \in V$ and a length bound $\ell \in \mathbb{N}$. An edge set $C \subseteq E$ is an \emph{$\ell$-length bounded $s$-$t$-cut} if all $s$-$t$-paths in $(V,E\setminus C)$ have length greater than $\ell$. The set $C$ is a \emph{minimum $\ell$-length bounded $s$-$t$-cut} if it minimizes the sum of the capacities of edges in $C$. When $\ell \geq |V(G)|-1$, the problem is equivalent to the standard minimum cut problem and can be solved via linear programming due to the well known max-flow min-cut theorem \citep{FordFu56}. However, for general $\ell \in \mathbb{N}$, \citet{AdamekKo71} showed that a generalization of this theorem does not hold. More recently, it was shown by \citet{BaierEr10} that finding a minimum $\ell$-length bounded $s$-$t$-cut is NP-hard for $\ell \in \{4, \dots, n^{1-\epsilon}\}$ for fixed $\epsilon >0$, even if capacities are uniform. By contrast, for $\ell \leq 3$, \citet{MahjoubMc10} showed that there exists a polynomial time algorithm which reduces the problem to a standard cut problem. In the following, we show how we can adjust and apply these results to our setting.

Despite its similarity to our problem 
(which can be observed by setting $G=T$, $u(e)=w(e), \ell=k$, and $s=x$), 
the problem described above differs in three ways from the problem under consideration. First, the node which should be disconnected, in this case $t$, is specified; second, edges are removed instead of reversed; and third, the graph is not restricted to be a tournament. For ease of presentation, we define a new problem which lies in between MoV for $k$-kings winners and minimum $\ell$-length bounded $s$-$t$-cuts. 

For a network $G=(V,E)$, we say that $C \subseteq E$ is an \emph{$\ell$-length bounded $s$-cut} if it is a $\ell$-length bounded $s$-$t$-cut for some $t \in V \setminus \{s\}$. We say that $C$ is a \emph{minimum $\ell$-length bounded $s$-cut}, if it is a minimum $\ell$-length bounded $s$-$t$-cut and capacity-minimizing among all $t \in V \setminus \{s\}$. 
Computing a minimum $\ell$-length bounded $s$-cut can be reduced to computing a minimum $\ell$-length bounded $s$-$t$-cut by iterating over all $t \in V \setminus \{s\}$.

The following lemma formalizes the connection between length bounded cuts and DRSs for $k$-kings. Though intuitive, note that the statement is not obvious because reversing the edges of a cut may create new paths of bounded length.\footnote{In the appendix, we give an example showing that Lemma 4  does not hold for cuts that are not minimal.} We define $\mathcal{P}_{x,y}(k)$ to be the set of $x$-$y$-paths in tournament $T$ of length at most $k$.

\begin{restatable}{lemma}{restateXcutKkingsDestr}\label{lem:xcut-kkings-destr}
A set $R \subseteq E(T)$ is a minimum DRS for $x$ w.r.t. $k$-kings iff $R$ is a minimum $k$-length bounded $x$-cut in $T$. 
\end{restatable}

%\restateXcutKkingsDestr*

\begin{proof}
The proof is divided into two parts. First, we show that every destructive reversal set is also an $x$-cut of equal cost, and second, we prove that every minimum $x$-cut is also a destructive reversal set of equal cost. This suffices to prove the claim.

Let $R \subseteq E(T)$ be a \drs for $x$. Assume for contradiction that $R$ is not an $x$-cut. Recall that we assume $x$ to be a $k$-king in $T$. Hence, for every $y$ there exists a path in $P_y \in \mathcal{P}_{x,y}(k)$ such that $P_y \cap R = \emptyset$ and therefore $x$ can reach every $y$ in tournament $T^R$ in at most $k$ steps, a contradiction to the assumption that $R$ is a \drs. We conclude that $R$ is an $x$-cut.

For the second part we show a slightly stronger statement, namely that for every $y$, every minimum $x$-$y$-cut is also a \drs of equal cost. Clearly it follows that every minimum $x$-cut is a \drs of equal cost. Let $R \subseteq E(T)$ be a minimum $x$-$y$-cut for some $y \in V(T) \setminus \{x\}$. We claim that all paths from $x$ to $y$ in $T^{R}$ have length strictly greater than $k$, which would suffice to prove the claim. Assume for contradiction that there exists an $x$-$y$-path $P$ in $T^{R}$ with $|P| \leq k$. Recall that $\overline{R}$ is the reversed counterpart of $R$ and note that $P \cap \overline{R} \neq \emptyset$, since otherwise this contradicts our choice of $R$ and $y$. Let $\{\overline{e}_1, \dots, \overline{e}_\ell\} := P \cap \overline{R}$ such that $\overline{e}_i$ appears before $\overline{e}_j$ in $P$ iff $i < j$. We label the connected components of $P \setminus \overline{R}$ such that $P_0$ is the subpath of $P$ from $x$ to the tail of edge $\overline{e}_1$, $P_\ell$ is the subpath from the head of $\overline{e}_\ell$ to $y$, and for $i \in \{1,\dots, \ell-1\}$, $P_i$ is the subpath starting at the head of $\overline{e}_{i}$ and ending at the tail of $\overline{e}_{i+1}$. Note in particular that $P_i$ can be empty. For every $\overline{e}_i$ recall that $e_i$ is its non-reversed counterpart from the original tournament $T$. We define $\mathcal{Q}_i$ to be the set of $x$-$y$-paths in $T$ which contain edge $e_i$ and are not longer than $k$. We claim that for every $e_i, i \in \{1,\dots,\ell\}$ there exists at least one path $Q_i \in \mathcal{Q}_i$ with $Q_i \cap R = \{e_i\}$. In the contrary case, since $w(e_i)>0$, the set $R \setminus \{e_i\}$ would be a feasible $x$-$y$-cut of smaller cost, a contradiction to the minimality of $R$. See \Cref{fig:destructive-kkings} for an illustration of the situation.

From the existence of $Q_i, i \in \{1, \dots, \ell\}$ we derive the existence of $x$-$y$-paths which appear in both $T$ and $T^{R}$ and therefore need to be longer than $k$. To this end we define $Q_i^{(1)}$ to be the first part of $Q_i$ from $x$ up to the tail of edge $e_i$ and denote by $Q_i^{(2)}$ the second part of $Q_i$, i.e., the part from the head of $e_i$ to $y$.

It is easy to see that $P_0 \dot\cup Q_{1}^{(2)}$, $Q_{\ell}^{(1)} \dot\cup P_\ell$ as well as $Q_{i}^{(1)} \dot\cup P_i \dot\cup Q_{i+1}^{(2)}, \; \forall \; i \in \{1, \dots, \ell-1\}$ are edge progressions from $x$ to $y$, i.e., sequences of edges leading from $x$ to $y$ which might use nodes and edges multiple times. In particular these multisets contain $x$-$y$-paths which do not intersect with $R$ and are therefore existent in both $T$ and $T^{R}$. We obtain 
\begin{align}
|P_0| + |Q_{1}^{(2)}| & > k \label{con:kkings:dest:ub1}\\ 
|Q_{\ell}^{(1)}| + |P_\ell| & > k\label{con:kkings:dest:ub2}\\ 
|Q_{i}^{(1)}| + |P_i| + |Q_{i+1}^{(2)}| & > k && \forall \; i \in \{1, \dots, \ell-1\}. \label{con:kkings:dest:ub3}
\end{align}
On the other hand, we know that $|P| \leq k$ and $|Q_i| \leq k$ for all $i \in \{1,\dots,\ell\}$ and hence 
\begin{align}
\sum_{i = 0}^{\ell} |P_i| + \ell & \leq k \label{con:kkings:dest:lb1} \\ 
|Q_{i}^{(1)}| + 1 + |Q_{i}^{(2)}| & \leq k && \forall \; i \in \{1, \dots, \ell\}. \label{con:kkings:dest:lb2}
\end{align}
Summing up \cref{con:kkings:dest:ub1,con:kkings:dest:ub2,con:kkings:dest:ub3} and \cref{con:kkings:dest:lb1,con:kkings:dest:lb2} results in 
\begin{align*}
k(\ell+1) &\stackrel{(\ref{con:kkings:dest:ub1})-(\ref{con:kkings:dest:ub3})}{<}&&\sum_{i=0}^\ell |P_i| + \sum_{i=1}^\ell (|Q_{i}^{(1)}| +  |Q_{i}^{(2)}|) \\
& \;\;\; < &&\sum_{i=0}^\ell |P_i| + \sum_{i=1}^\ell (|Q_{i}^{(1)}| +  |Q_{i}^{(2)}|) + 2\ell\\& \stackrel{(\ref{con:kkings:dest:lb1}), (\ref{con:kkings:dest:lb2})}{\leq} &&k(\ell+1),
\end{align*}
a contradiction to the existence of $P$. 
\end{proof}

Since there exist polynomial time algorithms for computing minimum $\ell$-length bounded $s$-$t$-cuts for $\ell \leq 3$ and $\ell=n-1$ \cite{MahjoubMc10,FordFu56}, Lemma \ref{lem:xcut-kkings-destr} immediately yields polynomial time algorithms for the minimum $\ell$-length bounded $s$-cut problem for $\ell \in \{2,3,n-1\}$.

\begin{corollary}
\label{cor:MoV-UC-kings}
Computing the MoV of a \uc winner, a $3$-king or a \tc winner in the weighted setting can be done in polynomial time.
\end{corollary}

The following result is obtained by carefully adjusting the proof of \citet{BaierEr10} showing that approximating minimum $\ell$-length bounded cuts for $\ell \geq 4$ is NP-hard. We give the entire proof in the appendix, where we also point out deviations from the original construction. 

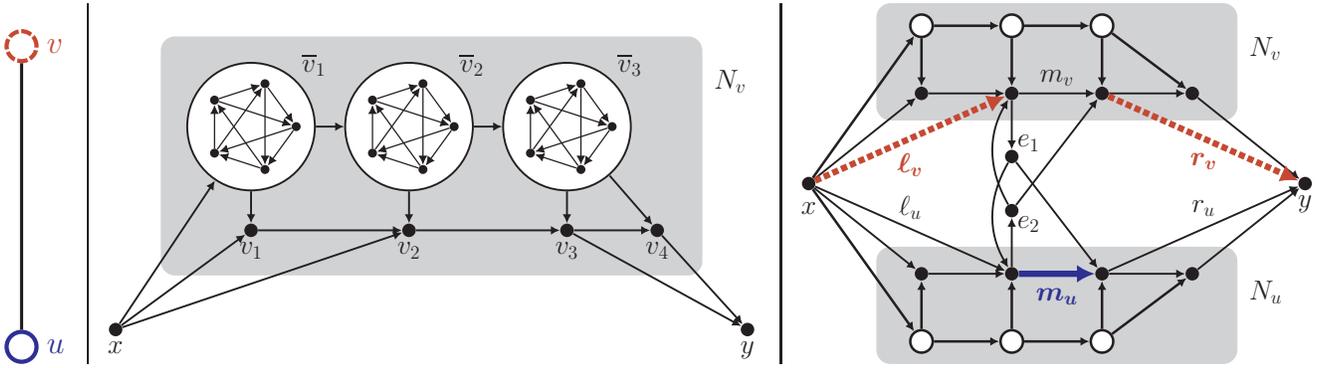
\begin{figure*}[t!]
\centering

\begin{tikzpicture}
\node[ultra thick,draw=red!80!black,circle,fill=white,inner sep=4pt, dash pattern = on 4pt off 1pt](v) at (0,4) {};\node[text=red!80!black,right=6pt] at (v){\Large$v$};

\node[circle,ultra thick,draw=blue,fill=white,inner sep=4pt](u) at (0,0) {};\node[blue,right=6pt] at (u){\Large$u$};

\draw[very thick] (u) -- (v);
\end{tikzpicture}\hspace{.1cm}
\vline \hspace{.1cm} \raisebox{0cm}{\scalebox{.6}{\begin{tikzpicture}

\fill [black!20, rounded corners=2ex] (-2,-1.3) rectangle (10,4);  \node[right=4pt] at (10,3){\huge$N_v$};

\node [circle, very thick, draw=black, fill=white, inner sep=1cm](d) at (0,2){}; \node[above right=1cm and 1cm] at (d){\huge$ \overline{v}_1$};
\node [circle, very thick, draw=black, fill=white, inner sep=1cm](d2) at (3.5,2){}; \node[above right=1cm and 1cm] at (d2){\huge$ \overline{v}_2$};
\node [circle, very thick, draw=black, fill=white, inner sep=1cm](d3) at (7,2){}; \node[above right=1cm and 1cm] at (d3){\huge$ \overline{v}_3$};

\node[circle,fill,inner sep=3pt](v1) at (0,-.3){}; \node[below=4pt] at (v1){\huge$ v_1$};
\node[circle,fill,inner sep=3pt](v2) at (3.5,-.3){}; \node[below=4pt] at (v2){\huge$ v_2$};
\node[circle,fill,inner sep=3pt](v3) at (7,-.3){}; \node[below=4pt] at (v3){\huge$ v_3$};
\node[circle,fill,inner sep=3pt](v4) at (9,-.3){}; \node[below=4pt] at (v4){\huge$ v_4$};

\node[circle,fill,inner sep=3pt](s) at (-3,-2.5){}; \node[below=4pt] at (s){\huge$x$};
\node[circle,fill,inner sep=3pt](t) at (11,-2.5){}; \node[below=4pt] at (t){\huge$y$};

\def \n {5}
\def \radius {1cm}
\def \margin {8} % margin in angles, depends on the radius 
\foreach \s in {1,...,\n}
{ 
      \node[circle,fill,inner sep=2pt,yshift=2cm](\s) at ({360/\n * (\s - 1)}:\radius) {}; 
}

\foreach \s in {6,...,10}
{ 
      \node[circle,fill,inner sep=2pt,xshift=3.5cm,yshift=2cm](\s) at ({360/\n * (\s - 1)}:\radius) {};
}

\foreach \s in {11,...,15}
{ 
      \node[circle,fill,inner sep=2pt,xshift=7cm,yshift=2cm](\s) at ({360/\n * (\s - 1)}:\radius) {};
}

\foreach \s/\t in {1/2,1/3,2/3,2/4,3/4,3/5,4/5,4/1,5/1,5/2,
			6/7,6/8,7/8,7/9,8/9,8/10,9/10,9/6,10/6,10/7,
			11/12,11/13,12/13,12/14,13/14,13/15,14/15,14/11,15/11,15/12}{
	\draw[<-,thick] (\s) -- (\t);}

\draw[->,very thick] (d) -- (v1);
\draw[->,very thick] (d) -- (d2);
\draw[->,very thick] (d2) -- (d3);
\draw[->,very thick] (d2) -- (v2);
\draw[->,very thick] (d3) -- (v3);
\draw[->,very thick] (d3) -- (v4);
\draw[->,very thick] (s) -- (v1);
\draw[->,very thick] (s) -- (d);
\draw[->,very thick] (v1) -- (v2);
\draw[->,very thick] (v2) -- (v3);
\draw[->,very thick] (v3) -- (v4);
\draw[->,very thick] (v4) -- (t);

\draw[->,very thick] (s) -- (v2);
\draw[->,very thick] (v3) -- (t);
 
\end{tikzpicture}}} \hspace{.1cm} \vline \hspace{.1cm} \scalebox{.6}{\begin{tikzpicture}
\tikzstyle{supernode}=[draw, ultra thick, circle,fill,inner sep=5pt,fill=white]
\tikzstyle{superarrow}=[ultra thick]

\fill [black!20, rounded corners=2ex] (1,1.4) rectangle (9,4); \node[right=4pt] at (9,3){\huge$N_v$};
\fill [black!20, rounded corners=2ex] (1,-1.4) rectangle (9,-4); \node[right=4pt] at (9,-2.4){\huge$N_u$};
%LEVEL 0 

\node[circle,fill,inner sep=3pt](s) at (-.5,0) {};\node[below=8pt] at (s){\huge$x$};

%LEVEL 1
\node[circle,fill,inner sep=3pt](v1) at (2,2) {}; \node[above right] at (v1){};%{\huge$v_1$};
\node[supernode](vb1) at (2,3.5) {}; \node[above=3pt] at (vb1){};%{\huge$\overline{v}_1$}; 
\node[circle,fill,inner sep=3pt](u1) at (2,-2) {}; \node[below right] at (u1){};%{\huge$ u_1$};
\node[supernode](ub1) at (2,-3.5) {}; \node[below=3pt] at (ub1){};%{\huge$ \overline{u}_1$};

%LEVEL 2
\node[circle,fill,inner sep=3pt](v2) at (4,2) {}; \node[above right] at (v2){};%{\huge$ v_2$};
\node[supernode](vb2) at (4,3.5) {}; \node[above=3pt] at (vb2){};%{\huge$ \overline{v}_2$}; 
\node[circle,fill,inner sep=3pt](u2) at (4,-2) {}; \node[below right] at (u2){};%{\huge$ u_2$};
\node[supernode](ub2) at (4,-3.5) {}; \node[below=3pt] at (ub2){};%{\huge$ \overline{u}_2$};

\node[circle,fill,inner sep=3pt](e1) at (4,.6) {}; \node[above right] at (e1){\huge$e_1$};
\node[circle,fill,inner sep=3pt](e2) at (4,-.6) {}; \node[below right] at (e2){\huge$e_2$};

%LEVEL 3
\node[circle,fill,inner sep=3pt](v3) at (6,2) {}; \node[above right] at (v3){};%{\huge$ v_3$};
\node[supernode](vb3) at (6,3.5) {}; \node[above=3pt] at (vb3){};%{\huge$ \overline{v}_3$}; 
\node[circle,fill,inner sep=3pt](u3) at (6,-2) {}; \node[below right] at (u3){};%{\huge$ u_3$};
\node[supernode](ub3) at (6,-3.5) {}; \node[below=3pt] at (ub3){};%{\huge$ \overline{u}_3$};

%LEVEL 4
\node[circle,fill,inner sep=3pt](v4) at (8,2) {}; \node[above right] at (v4){};%{\huge$ v_4$};
\node[circle,fill,inner sep=3pt](u4) at (8,-2) {}; \node[below right] at (u4){};%{\huge$ u_4$};

%LEVEL5
\node[circle,fill,inner sep=3pt](t) at (10.5,0) {};\node[below=4pt] at (t){\huge$y$};

%EDGES 
\draw[->,very thick] (s) -- (v1);
\draw[->,superarrow] (s) -- (vb1);

\draw[->,very thick] (s) -- (u1);
\draw[->,superarrow] (s) -- (ub1);
%\draw[->,very thick] (s) -- (u2);

\draw[->,superarrow] (vb1) -- (vb2);
\draw[->,superarrow] (vb2) -- (vb3);
\draw[->,superarrow] (vb3) -- (v4);
\draw[->,superarrow] (vb1) -- (v1);
\draw[->,superarrow] (vb2) -- (v2);
\draw[->,superarrow] (vb3) -- (v3);
\draw[->, very thick] (v1) -- (v2);
%\draw[->,very thick] (v2) -- (v3);
\draw[->,very thick] (v3) -- (v4);
%\draw[->,thick] (v3) -- (t);
\draw[->,very thick] (v4) -- (t);

\draw[->,superarrow] (ub1) -- (ub2);
\draw[->,superarrow] (ub2) -- (ub3);
\draw[->,superarrow] (ub3) -- (u4);
\draw[->,superarrow] (ub1) -- (u1);
\draw[->,superarrow] (ub2) -- (u2);
\draw[->,superarrow] (ub3) -- (u3);
\draw[->,very thick] (u1) -- (u2);
%\draw[->,thick] (u2) -- (u3);
\draw[->,very thick] (u3) -- (u4);
%\draw[->,very thick] (u3) -- (t);
\draw[->,very thick] (u4) -- (t);

\draw[->,very thick] (v2) -- (e1);
\draw[->,very thick] (u2) -- (e2);
\draw[->,very thick] (e2) to [bend left](v2);
\draw[->,very thick] (e1) to [bend right](u2);
\draw[->,very thick] (e1) -- (u3);
\draw[->,very thick] (e2) -- (v3);

\draw[->,line width=1.4mm,red!80!black,>={Latex[width=4mm,length=4mm]},dash pattern = on 4pt off 2pt] (s) -- node[below=4pt]{\huge$\boldsymbol{\ell_v}$}(v2); %lv
\draw[->,very thick,black] (v2) -- node[above=1pt]{\huge${m_v}$}(v3); %mv 
\draw[->,line width=1.4mm,red!80!black,>={Latex[width=4mm,length=4mm]},dash pattern = on 4pt off 2pt] (v3) -- node[below=4pt]{\huge$\boldsymbol{r_v}$}(t); %rv 

\draw[->,very thick,black] (s) -- node[above=2pt]{\huge${\ell_u}$}(u2); %lu 
\draw[->,line width=1.4mm,blue,>={Latex[width=4mm,length=4mm]}] (u2) -- node[below=4pt]{\huge$\boldsymbol{m_u}$}(u3); %mu 
\draw[->,very thick,black] (u3) -- node[above=4pt]{\huge${r_u}$}(t); %ru 

%LABELS

\end{tikzpicture}}

\caption{Illustration of the construction used in the proof of \Cref{thm:destr:kkings:nphard} for the case $k=4$. For any graph $G$ (left image), a tournament $T$ is constructed by introducing node gadgets and edge gadgets as follows. A \emph{node gadget} $N_v$ consists of four nodes $v_1, v_2, v_3, v_4$ and three supernodes $\overline{v}_1, \overline{v}_2, \overline{v_3}$, where the latter are tournaments themselves. The center image shows the node gadget for node $v$. An \emph{edge gadget} for $e=\{u,v\}$ consists of two nodes $e_1, e_2$ and edges connecting the node gadgets of $u$ and $v$; see the right image. Nodes $x$ and $y$ are connected to all node gadgets as illustrated. All omitted edges point ``backwards'' (from right to left) and the direction of vertical edges, if not specified, can be chosen arbitrarily. 
}
\label{fig:kkingsReduction}
\end{figure*}

\begin{restatable}{theorem}{restateKKingsReduction}
For any constant $k\geq 4$, computing the MoV of a $k$-king in the unweighted setting is NP-hard. For any constant $\epsilon>0$, the problem is still NP-hard when we restrict to non-constant $k \geq n^{1-\epsilon}$. \label{thm:destr:kkings:nphard}
\end{restatable}

\begin{proofsketch}
We reduce from vertex cover; see \Cref{fig:kkingsReduction} for the construction for $k=4$. Lemma \ref{lem:xcut-kkings-destr} implies that determining the MoV of node $x$ with respect to $4$-kings is equivalent to computing the cost of a $4$-length bounded minimum $x$-cut. The key part of the proof is to show that, for any $c \leq |V(G)|$, there exists a vertex cover in $G$ of size $c$ iff there exists a $4$-bounded $x$-cut in $T$ of size $c + |V(G)|$. For the direction from left to right, a vertex cover $U$ can be translated to a $4$-bounded $x$-$y$-cut by including edges $\ell_v$ and $r_v$ (depicted by red dashed edges) whenever $v \in U$ (depicted by a red dashed node) and $m_v$ otherwise. For the other direction, we argue that any $4$-bounded $x$-cut of size $c + |V(G)|$ can be translated to a $4$-bounded $x$-$y$-cut which includes only edges of the type $\ell_v,r_v$ and $m_v$ and is of no greater size. Reversing the previously described transformation gives us a vertex cover of size $c$. The proof is then extended to $k>4$.
\end{proofsketch}

\subsection{Banks Set}

Deciding whether an alternative $x$ is contained in the Banks set of a tournament $T$, and hence deciding whether $\mov_\ba(x,T)>0$, is NP-complete \citep{Woeginger03}.
Our next result shows that determining $\mov_\ba(x,T)$ is computationally intractable \textit{even} if we know that $x$ is a Banks winner in tournament $T$.
% (\ie $\mov_\ba(x,T)>0$).

\begin{theorem}
\label{thm:destructive-banks-NP}
Computing the MoV of a \ba winner in the unweighted setting is NP-hard.
\end{theorem}

\begin{proof}
We reduce from the NP-hard problem of determining whether an alternative is contained in the Banks set \citep{Woeginger03}.
Take any instance of that problem, which consists of a tournament $T$ and one of its alternatives $x$.
Add two alternatives $y,z$ so that $y$ dominates only $\overline{D}(x)\cup \{z\}$, and $z$ dominates only $D(x)$.
Call the resulting tournament~$T'$ (see Figure~\ref{fig:destructive-banks}).
Observe that $x\in\ba(T')$: the transitive subtournament $T|_{\{x,y\}}$ cannot be extended, since no alternative dominates both $x$ and $y$.
We claim that $\mov_\ba(x,T')=1$ if and only if $x\not\in\ba(T)$.

First, assume that $x\not\in\ba(T)$. We show that $R=\{x,y\}$ is a DRS for $x$. 
Consider any transitive subtournament in $T'' = (T')^R$ with $x$ as the maximal element.
This tournament cannot include $y$, but may include $z$.
Since $x\not\in\ba(T)$, there exists an alternative $w$ in $T$ that dominates all alternatives in the subtournament. In particular, since $w\in \overline{D}(x)$, $w$ also dominates $z$.
Hence the transitive subtournament can be extended by $w$, implying that $x\not\in\ba(T'')$.

Assume now that $x\in\ba(T)$. We claim that $\mov_\ba(x,T')>1$.
Since $x\in\ba(T)$, there exists a transitive subtournament in $T$ with $x$ as the maximal element that cannot be extended by any alternative in $T$.
Moreover, since $x$ dominates both $y$ and $z$, this subtournament cannot be extended by $y$ or $z$.
Unless we reverse an edge in $T$ or the edge $(x,z)$, this subtournament still cannot be extended.
If we reverse the edge $xz$, the transitive subtournament $T|_{\{x,y\}}$ cannot be extended.
Else, if we reverse an edge in $T$, the transitive subtournament $T|_{\{x,y,z\}}$ cannot be extended.
Hence, there is no DRS for $x$ of size one, as claimed. 
\end{proof}

\begin{figure}[!ht]
\centering
\scalebox{1.1}{
\begin{tikzpicture}
\draw[dashed] (2,2) ellipse (3cm and 1cm);
\draw (0.6,2) circle [radius = 0.6];
\node at (0.6,2) {$\overline{D}(x)$};
\draw (3.4,2) circle [radius = 0.6];
\node[circle,fill,inner sep=2pt](x) at (2,2){}; \node[below=2pt] at (x){$x$};
\node at (3.4,2) {$D(x)$};
\draw[->,thick] (1.2,2) -- (1.8,2);
\draw[->,thick] (2.2,2) -- (2.8,2);
\node at (4.5,2.9) {$T$};
\node[circle,fill,inner sep=2pt](y) at (0.6,0){}; \node[below=2pt] at (y){$y$};
\node[circle,fill,inner sep=2pt](z) at (3.4,0){}; \node[below=2pt] at (z){$z$};
\draw[->,thick] (0.6,0.2) -- (0.6,1.4);
\draw[->,thick] (3.4,0.2) -- (3.4,1.4);
\draw[->,thick] (0.8,0) -- (3.2,0);
\draw[->,thick] (1.85,1.85) -- (0.67,0.15);
\draw[->,thick] (2.15,1.85) -- (3.33,0.15);
\draw[->,thick] (1.05,1.6) -- (3.23,0.1);
\draw[->,thick] (2.95,1.6) -- (0.77,0.1);
\end{tikzpicture}
}
\caption{Illustration of the tournament $T'$ constructed in the proof of Theorem~\ref{thm:destructive-banks-NP}.}
\label{fig:destructive-banks}
\end{figure}
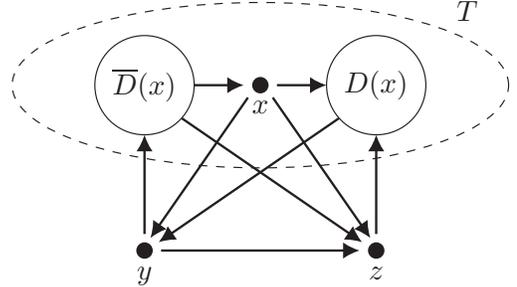

\section{Computing the MoV for Non-Winners}\label{sec:compNonWinners}

We now turn to computing the MoV for non-winners.

\subsection{Copeland}
Similarly to the winner case, the results by \citet{FHHR09c} already imply that the MoV for non-winners can be computed in polynomial time. For completeness, we remark that a greedy algorithm suffices for our unweighted setting, and present an easy network flow approach for the weighted case. 
 
\begin{restatable}{theorem}{restateCopelandConstructive}\label{thm:constructive-CL}
Computing the MoV of a \cp non-winner in the weighted setting can be done in polynomial time.
\end{restatable}

\subsection{Uncovered Set, $\boldsymbol{k}$-Kings, and Top Cycle}

To get $x$ into the uncovered set, we need its dominion to be a dominating set in $T_{-x}$.
Since a tournament with $n$ vertices always has a dominating set of size $\lceil \log_2 n\rceil$ \citep{MegiddoVi88}, we do not need to flip more than $\lceil \log_2 n\rceil$ edges. 
This also means that there exists an $n^{O(\log n)}$ algorithm for finding the minimum number of necessary edge reversals, as we can try all combinations of at most $\lceil \log_2 n\rceil$ vertices to add to the dominion of $x$.
\citet{MegiddoVi88} also proved that the problem of finding a dominating set of minimum size in a tournament, which we call {\normalfont \scshape Minimum Dominating Set}, is unlikely to admit a polynomial-time algorithm: the existence of such an algorithm would have unexpected implications on the satisfiability problem.
We present a reduction from {\normalfont \scshape Minimum Dominating Set} to the problem of computing the MoV for \uc non-winners, which means that the latter problem is also unlikely to admit an efficient algorithm.

\begin{theorem}
\label{thm:constructive-UC-unweighted}
Computing the MoV of a \uc non-winner in the unweighted setting is at least as hard as {\normalfont \scshape Minimum Dominating Set}.
\end{theorem}

\begin{proof}
Consider an instance of {\normalfont \scshape Minimum Dominating Set} given by a tournament $T$.
Define a new tournament $T'$ by adding an alternative $x \notin V(T)$ to $T$, and by making $x$ a Condorcet loser in $T'$.
We claim that 
$-\mov_\uc(x,T')$
is equal to the minimum size of a dominating set in $T$.
For any dominating set in $T$, we obtain a constructive reversal set for $x$ in $T'$ consisting of the edges between $x$ and all members of the set.
On the other hand, consider a CRS $R$ for $x$ in $T'$.
Suppose that $R$ contains an edge $(z,y)$ with $x\not\in\{z,y\}$, such that $y\rightarrow z$ in $T'^R$. 
The only alternative that this reversal can help $x$ reach in two steps is $z$.
In this case, we can instead include $(z,x)$ in $R$ and maintain the property that $x$ can reach all other alternatives in at most two steps.
Hence there is always a minimal CRS that only contains edges incident to $x$.
The alternatives involved in this CRS besides $x$ form a dominating set in $T$.
\end{proof}

In the unweighted setting, minimum CRSs  w.r.t. $k$-kings ($k\geq3$) are single edges (see \Cref{thm:constructive-bound-tc} in \Cref{subsec:lowerBounds}) and hence can be found efficiently.
In the weighted setting, we show hardness for \uc and $k$-kings and tractability for~\tc.

\begin{restatable}{theorem}{restateThmConstructiveUCNP}
\label{thm:constructive-UC-NP}
Computing the MoV of a \uc non-winner in the weighted setting is NP-hard.
\end{restatable}

\begin{figure}[!b]
\centering
\vspace{-.5cm}
\scalebox{1.1}{
\begin{tikzpicture}
\node[circle,fill,inner sep=2pt] at (2,8)(x){}; \node[above=2pt] at (x){$x$};
\node[circle,fill,inner sep=2pt] at (2,7)(y){}; \node[right=2pt] at (y){$y$};
\draw[->,thick] (2,7.8) -- (2,7.2);
\draw[->,thick] (2,6.8) -- (2,6.25);
\draw (2,5.65) ellipse (1.4cm and 0.6cm);
\draw (2,4) ellipse (1.8cm and 0.6cm);
\node at (3.7,5.65) {$A$};
\draw[fill] (1.3,5.65) circle [radius = 0.05];
\draw[fill] (2,5.65) circle [radius = 0.05];
\draw[fill] (2.7,5.65) circle [radius = 0.05];
\node at (4.1,4) {$B$};
\draw (1.3,4) ellipse (0.6cm and 0.3cm);
\draw (2,4) ellipse (0.6cm and 0.3cm);
\draw (2.7,4) ellipse (0.6cm and 0.3cm);
\draw[->,thick] (1.3,5.65) -- (1.3,4.3);
\draw[->,thick] (2,5.65) -- (2,4.3);
\draw[->,thick] (2.7,5.65) -- (2.7,4.3);
%\draw[->,thick] (2.8,6.15) -- (2.13,7.8);
\draw[->,thick] (0.19,4) to[bend left=50] (1.85,7.9);
\draw[->,thick] (0.19,4) to[bend left=50] (1.85,7);
\draw[->,thick] (3.4,5.75) to[bend right=35] (2.15,7.9);
\end{tikzpicture}
}
\caption{An illustration of the construction in Theorem~\ref{thm:constructive-UC-NP}.}
\label{fig:constructive-UC}
\end{figure}
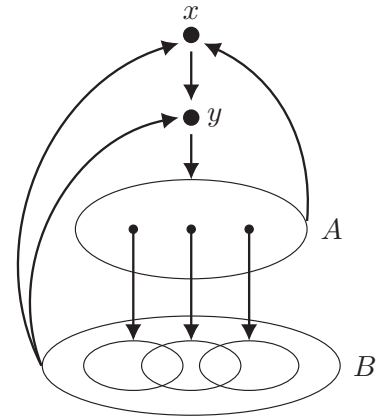

%\restateThmConstructiveUCNP*
\begin{proof}
Given an instance of {\normalfont \scshape Set Cover} with a universe of size $r$ and a collection of $s$ sets, we construct a tournament with alternatives $x, y$ and alternative sets $A, B$, where $|A| = s$ and $|B| = r$. The edges are given as follows (see also Figure~\ref{fig:constructive-UC}):
\begin{itemize}
\item $A\succ x\succ y\succ A$;
\item $B\succ \{x,y\}$;
\item Each alternative in $A$ dominates the corresponding subset of $B$ in the {\normalfont \scshape Set Cover} instance, and is dominated by the remaining alternatives in $B$.
\end{itemize}
The edges within $A$ and $B$ are arbitrary. 
The edges between $A$ and $x$ have cost $1$, while the remaining edges have cost~$n^2$.

The chosen costs imply that a miminum CRS will only contain edges between $A$ and $x$.
Since $x$ already reaches all alternatives of $A$ in two steps via $y$, it only needs to reach all vertices of $B$ in two steps via $A$ in order to be part of \uc. 
The minimum cost of a CRS is therefore exactly the size of a minimum set cover.
\end{proof}

\begin{restatable}{theorem}{restateReductionKKingsConstructive}
For any constant $k\geq 3$, computing the MoV of a non-$k$-king in the weighted setting is NP-hard. For any constant $\epsilon>0$, the problem is still NP-hard when we restrict to non-constant $k \geq (1-\epsilon)n$. \label{thm:constructive-weighted-kkings-NP}
\end{restatable}

%\restateReductionKKingsConstructive*

\begin{proof}
The proof is divided into two parts. First, we introduce the reduction for any constant $k\geq3$. Second, we argue that, for any $\epsilon >0$, we can still carry out the reduction even when we restrict ourselves to the problem in which $k\geq (1-\epsilon)n$. 

We use a similar reduction as for \uc. Instead of having a single alternative $y$, we add $k-1$ alternatives $y_1,\dots,y_{k-1}$ so that
\begin{itemize}
\item $x\succ y_1\succ\dots\succ y_{k-1}\succ A$;
\item $y_i\succ x$ for $i\geq 2$;
\item $y_j\succ y_i$ for $j\geq i+2$;
\item $A\succ\{x,y_1,\dots,y_{k-2}\}$;
\item $B\succ\{x,y_1,\dots,y_{k-2},y_{k-1}\}$;
\item Each alternative in $A$ dominates the corresponding subset of $B$ in the {\normalfont \scshape Set Cover} instance, and is dominated by the remaining alternatives in $B$.
\end{itemize}
The edges within $A$ and $B$ are arbitrary. 
The edges between $A$ and $y_{k-2}$ have cost $1$, while the remaining edges have cost $n^2$.

The choice of edge costs implies that a minimum CRS only contains edges between $A$ and $y_{k-2}$.
Since $x$ already reaches all alternatives of $A$ in $k$ steps via $y_1,\dots,y_{k-1}$, it only needs to reach all vertices of $B$ in $k$ steps via $y_1,\dots,y_{k-2},A$ in order to be part of the uncovered set. 
The minimum cost of a CRS is therefore exactly the size of a minimum set cover.

It remains to argue that even if we restrict ourselves to the problem with $k \geq (1-\epsilon)n$, for any fixed $\epsilon>0$, we can still carry out the above reduction in polynomial time. Let $\epsilon >0$ be given and let $r$ be the size of the universe and $s$ be the number of sets in the {\normalfont \scshape Set Cover} instance. We define $n_k$ to be the number of nodes in the tournament which we construct for a given $k$. More precisely it holds that \[n_k = k + s + r.\]
Choose the smallest $k \in \mathbb{N}_{\geq 3}$ such that \[k \geq (s+r)\left(\frac{1-\epsilon}{\epsilon}\right).\] This is still polynomial in $s$ and $r$ while it implies that \[k = (1-\epsilon)\left(k + \frac{\epsilon}{1-\epsilon}k\right) \geq (1-\epsilon)(k+s+r) \geq (1-\epsilon )n_k,\] concluding the proof.
\end{proof}

\begin{restatable}{theorem}{restateThmConstructiveWeightedTC}
Computing the MoV of a \tc non-winner in the weighted setting can be done in polynomial time. \label{thm:constructive-weighted-tc}
\end{restatable}

%\restateThmConstructiveWeightedTC*

\begin{proof}
Consider a partition of $T$ into strongly connected components. 
These components form a linear order with all vertices in an earlier component dominating all vertices in a later component.
Call the components $T_1,\dots,T_k$ according to the linear order, and assume that $x$ belongs to $T_r$. Since $x \notin \tc(T)=V(T_1)$, we have $r \ge 2$. 
Construct a tournament $T'$ with vertices $v_1,\dots,v_r$.
For $1\leq i < j\leq r-1$, add a directed edge $v_j\rightarrow v_i$ with cost equal to the minimum cost of an edge between an alternative in $T_i$ and an alternative in $T_j$.
For $1\leq i\leq r-1$, add a directed edge $v_r\rightarrow v_i$ with cost equal to the minimum cost of an edge between an alternative in $T_r\cup T_{r+1}\cup\dots\cup T_k$ and an alternative in $T_i$.

We claim that the shortest path distance from $T_r$ to $T_1$ in~$T'$ equals the minimum cost of a CRS for $x$ in $T$. 
Take any shortest path from $T_r$ to $T_1$. For each edge on this path, we reverse a corresponding edge in $T$ with the same cost.
This allows $x$ to reach all components on this path, including $T_1$.
Note that $x$ can already reach the components $T_{r+1},\dots,T_k$ even before the reversals.
Moreover, the remaining components are directly reachable from $T_1$, and therefore $x$ can also reach them.\footnote{An exception to this is if we reverse an edge $v_p\rightarrow v_1$, and both $T_p$ and $T_1$ are singletons. However, in this case $x$ can reach $T_p$ via the shortest path.}
Hence we can bring $x$ into \tc using no more cost than that of the shortest path.
On the other hand, any CRS for $x$ must have the effect that $x$ can reach an alternative in $T_1$.
This gives rise to a path from $T_r$ to $T_1$ in $T'$ with no greater cost.

Computing strongly connected components of $T$ can be done in time linear in the input size using Tarjan's algorithm or Kosaraju's algorithm, and finding the shortest path can be done in polynomial time using Dijkstra's algorithm. Therefore our algorithm runs in polynomial time.
\end{proof}

\subsection{Banks Set}

For Banks non-winners, we present an analogous result as in the winner case: even if we know that $x$ has a negative MoV in tournament $T$, determining $\mov_\ba(x,T)$ is intractable.

\begin{restatable}{theorem}{restateThmConstructiveBanksNP}
\label{thm:constructive-banks-NP}
Computing the MoV of a \ba non-winner in the unweighted setting is NP-hard.
\end{restatable}

%\restateThmConstructiveBanksNP*

\begin{figure}[!t]
\centering
\scalebox{1.2}{
\begin{tikzpicture}
\draw[dashed] (0.9,2) ellipse (3.4cm and 1cm);
\draw (2.7,2) circle [radius = 0.7];
\node[circle,fill,inner sep=2pt] at (0.9,2)(x){}; \node[below=2pt] at (x){$x$};
\node at (2.7,2) {$D(x)$};
\draw[->,thick] (1.1,2) -- (2,2);
\draw[->,thick] (-0.2,2) -- (0.7,2);
\node[circle,fill,inner sep=2pt] at (0,4)(y){}; \node[above=2pt] at (y){$y$};
\node[circle,fill,inner sep=2pt] at (1.8,4)(z){}; \node[above=2pt] at (z){$z$};
\node at (3.8,2.9) {$T$};
\draw[->,thick] (0.2,4) -- (1.6,4);
\draw[->,thick] (0.1,3.8) -- (0.9,2.2);
\draw[->,thick] (0.15,3.9) -- (2.05,2.25);
\draw[->,thick] (1.8,3.8) -- (1.8,2.98);
\draw[->,thick] (-0.6,2.65) -- (-0.05,3.8);
\draw (-0.9,2) circle [radius = 0.7];
\node at (-0.9,2) {$\overline{D}(x)$};
\end{tikzpicture}
}

\caption{An illustration of the construction in Theorem~\ref{thm:constructive-banks-NP}.}
\label{fig:constructive-banks}
\end{figure}
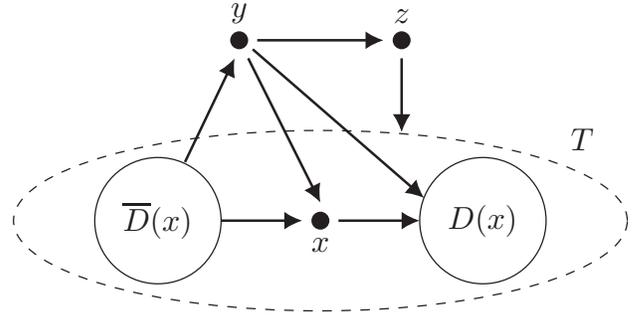

%todo:changeback \citeAppendix{Woeginger03}
\begin{proof}
We reduce from the NP-hard problem of determining whether an alternative is contained in the Banks set \cite{Woeginger03}.
Take any instance of that problem, which consists of a tournament $T$ and one of its alternative $x$.
Add two alternatives $y,z$ so that $y$ dominates $D(x)\cup \{x,z\}$, but is dominated by $\overline{D}(x)$, while $z$ dominates all alternatives in $T$.
Call the resulting tournament $T'$ (see Figure~\ref{fig:constructive-banks}).
Observe that $x\not\in\ba(T')$: any transitive subtournament with $x$ as the maximal element cannot contain $y$, and can therefore be extended by $z$.
Therefore, $\mov_\ba(x,T')<0$.
We claim that 
$\mov_\ba(x,T')=-1$
if and only if $x\in\ba(T)$.

First, assume that $x\in\ba(T)$. 
This means $x$ is the maximal element of a transitive subtournament $T''$ of $T$ that cannot be extended by any alternative in $T$.
Reverse the edge $xy$, and insert $y$ into $T''$ at the position after $x$.
The resulting subtournament cannot be extended by any alternative in $T$, nor can it be extended by $z$ because $y\succ z$. 
Hence $x\in\ba(T')$ after the reversal.

Assume now that $x\not\in\ba(T)$, and suppose for contradiction that $\mov_\ba(x,T')=-1$.
% it is possible to make $x$ a Banks winner with one edge reversal.
Notice that $x$ is covered by both $y$ and $z$, so in order to get $x$ into $\ba(T')$, we need to either strengthen $x$, or weaken both $y$ and $z$.
The latter option cannot be accomplished with one reversal, so the reversal needs to strengthen $x$.
If it strengthens $x$ against another alternative in $T$, $x$ is still covered by $z$. 
If it strengthens $x$ against $z$, $x$ is still covered by $y$. 
Hence the reversal must strengthen $x$ against $y$. 
Consider any transitive subtournament $T''$ %of $T'$
after the reversal with $x$ as the maximal element. 
The tournament $T''$ cannot contain $z$, but may contain $y$. 
However, since $y$ has the same dominion as $x$ in $T$, an alternative that extends $T''_{-y}$, which must exist because $x\not\in\ba(T)$, necessarily extends $T''$. 
This yields the desired contradiction.
\end{proof}

\section{Bounds on the Margin of Victory}

In this section, we consider the unweighted setting and establish bounds on the MoV values for winners and non-winners.
There are at least two insights that one could draw from these bounds. First, tournament solutions with a low absolute value of MoV bound yield manipulability guarantees; indeed, if the absolute value of the MoV bound is low, then a manipulator can always obtain the desired outcome by reversing a small number of edges regardless of the tournament instance. Second, knowing these bounds is useful for understanding the actual MoV for specific tournaments. For example, one can calculate the ``relative/normalized MoV'' by dividing the actual MoV value by the bound. The resulting ratio provides a measure of how far away an alternative is from winning or losing; in contrast to the standard MoV measure, the relative MoV also enables us to make comparisons between tournaments of different sizes.

\subsection{Upper Bounds for Winners}

We show that for all considered tournament solutions, one may need to reverse up to $\lfloor n/2\rfloor$ edges to take a winner out of the winner set, but no more.

\begin{theorem}
Let $S\in\{\cp, \tc, \uc, \ba, k\text{-kings}\}$, where $k\geq 3$. % is arbitrary.
For any tournament $T$ and any $x\in S(T)$, we have $\mov_S(x,T)\leq\lfloor n/2\rfloor$.
Moreover, this bound is tight. \label{thm:destructive-bounds}
\end{theorem}

\begin{proof}

Since all of the tournament solutions considered are contained in \tc, an upper bound for \tc carries over to the other solutions as well. By analogous reasoning, it suffices to show the tightness of the bound for \ba and \cp. 

We first prove the upper bound.
Let $y$ be an arbitrary Copeland winner in $T_{-x}$.
Since $T_{-x}$ consists of $n-1$ alternatives, $y$ dominates at least $\lceil(n-2)/2\rceil = \lceil n/2\rceil - 1$ other alternatives.
Hence, we can make $y$ a Condorcet winner in $T$ by reversing at most $(n-1)-(\lceil n/2\rceil - 1) = \lfloor n/2\rfloor$ edges.
Since \tc is Condorcet-consistent, $\lfloor n/2\rfloor$ edge reversals suffice to take $x$ out of \tc.

Next, we show the lower bound for \ba. Assume first that $n$ is even, say $n=2\ell$.
Besides $x$, suppose that $T$ contains alternatives $y_1,\dots,y_{2\ell-1}$, which are placed around a circle in clockwise order.
Each alternative dominates the $\ell-1$ following alternatives in clockwise order (e.g., $y_1$ dominates $y_2,\dots,y_\ell$), and all $2\ell-1$ alternatives are dominated by $x$.
We claim that taking $x$ out of the Banks set requires at least $\lfloor n/2\rfloor = \ell$ edge reversals.
Consider the $2\ell-1$ sets 
\begin{align*}
&\{y_1,y_\ell\},\{y_2,y_{\ell+1}\},\dots,\{y_\ell,y_{2\ell-1}\}, \\ 
&\{y_{\ell+1},y_1\}, \dots, \{y_{2\ell-1}, y_{\ell-1}\}.
\end{align*}
Note that each $y_i$ is contained in exactly two of these sets.
For each set, we say that it is `good' if the only alternative that dominates both of the alternatives in the set is $x$, and `bad' otherwise. 
Note that the existence of a good set implies that $x$ is a Banks winner, as the transitive subtournament consisting of the good set and $x$ cannot be extended.
Initially, all $2\ell-1$ sets are good.
A reversal involving $x$ and $y_i$ can turn at most two good sets into bad sets (i.e., the two sets containing $y_i$).
Similarly, a reversal involving $y_i$ and $y_j$, where $y_j$ dominates $y_i$ after the reversal, can make at most two good sets bad (i.e., the two sets containing $y_i$).
So after at most $\ell-1$ reversals, at least one set is still good.
This implies that there is no DRS of size at most $\ell-1$.
% Hence at least $\ell$ reversals are needed in order to take $x$ out of $\ba(T)$.
Hence, $\mov_\ba(x,T)\ge \ell = \lfloor n/2\rfloor$.

The case where $n$ is odd can be handled similarly.
Let $n=2\ell-1$.
Construct a tournament with alternatives $x,y_1,\dots,y_{2\ell-1}$ as before, and remove $y_{2\ell-1}$.
We claim that taking $x$ out of the Banks set in this tournament requires at least $\lfloor n/2\rfloor = \ell-1$ edge reversals.
Consider $2\ell-3$ sets, starting with the $2\ell-1$ sets above and removing the two sets that contain $y_{2\ell-1}$.
Each $y_i$ is contained in at most two of these sets.
The previous argument can be applied to show that 
$\mov_\ba(x,T)\ge \ell-1 = \lfloor n/2\rfloor$.
% at least $\ell-1$ reversals are necessary to make $x$ a non-winner.

To conclude the proof, we show that the same tournaments as constructed above also imply the tightness of the bound $\lfloor n / 2 \rfloor$ for \cp. 
In order to make $x$ a non-winner, we must reverse edges so that another alternative $y$ has a larger dominion than $x$.
If $n$ is even, then initially $x$ dominates $n-1$ alternatives while $y$ dominates $n/2-1$ alternatives, so $|D(x)|-|D(y)| = (n-1)-(n/2-1) = n/2$.
Each edge reversal decreases this difference by at most $1$, except for the reversal of the edge $(x,y)$, which reduces the difference by $2$.
Hence, in order to make the difference negative, we need at least $n/2$ reversals.
A similar argument applies for the case where $n$ is odd, since we have $|D(y)|\leq (n-1)/2$, and therefore $|D(x)|-|D(y)| \geq (n-1)/2 = \lfloor n/2\rfloor$.
\end{proof}

\subsection{Lower Bounds for Non-Winners}\label{subsec:lowerBounds}

Next, we turn our attention to non-winners.
For \tc and $k$-kings with $k\geq 3$, it is clear that reversing one edge suffices to make any alternative a winner. Indeed, we can simply reverse the edge between $x$ and an arbitrary alternative in the uncovered set of $T_{-x}$. 
This ensures that $x$ can reach every other alternative via a directed path of length at most three.

\begin{theorem}
\label{thm:constructive-bound-tc}
Let $S\in\{\tc, k\text{-kings}\}$, where $k\geq 3$ is arbitrary. For any tournament $T$ and any $x\in V(T)\setminus S(T)$, we have $\mov_S(x,T) = -1$.
\end{theorem}

For \cp, as many as $n-2$ edge reversals may be required.

\begin{theorem}
\label{thm:constructive-bound-co}
For any tournament $T$ and any $x\in V(T)\setminus \cp(T)$, we have $\mov_\cp(x,T) \geq -(n-2)$. Moreover, this bound is tight.
\end{theorem}

\begin{proof}
With a budget of $n-2$ reversals, we can make $x$ dominate at least $n-2$ alternatives.
Moreover, if the tournament initially contains a Condorcet winner, one of these reversals can be used to make $x$ dominate it, meaning that every alternative dominates at most $n-2$ alternatives after the reversals.
Hence $x$ becomes a Copeland winner.

To show tightness, consider a tournament where $x$ is a Condorcet loser and there is a Condorcet winner $y$. We have $|D(y)|-|D(x)| = n-1$.
Each edge reversal reduces this difference by at most $1$, except for the reversal of the edge $(x,y)$, which reduces the difference by $2$.
In order for $x$ to be a Copeland winner, this difference must be nonpositive.
It follows that we need at least $n-2$ reversals, as claimed.
\end{proof}

Finally, we show that for \uc and \ba, reversing $O(\log n)$ edges can bring any alternative into the winner set. %, and this cannot be improved.

\begin{theorem}
\label{thm:constructive-bound-uc-ba}
Let $S\in\{\uc, \ba\}$. For any tournament $T$ and any $x\in V(T)\setminus S(T)$, we have $\mov_S(x,T) \geq -\lceil \log_2 n\rceil$.
Moreover, this bound is asymptotically tight.
\end{theorem}

\begin{proof}
Since $\ba\subseteq\uc$, it suffices to establish the bound for \ba and the tightness for \uc.
We first prove the bound for \ba, by considering a tournament $T$ and iteratively constructing a CRS for an alternative $x \notin \ba(T)$.
Let $T'$ be a transitive subtournament of $T$ that initially contains only the alternative $x$, and let $B$ be the set of alternatives that dominate all alternatives in $T'$.
Let $\ell = |B|$, and let $y$ be a Copeland winner of the tournament $T|_B$. 
% restricted to the alternatives in $B$.
Note that $y$ dominates at least $\lceil (\ell-1)/2\rceil$ other alternatives in $B$ as well as $x$.
We reverse the edge between $x$ and $y$, insert $y$ into the transitive tournament $T'$ at the position after $x$, and update the set $B$.
Since $y$ is added to $T'$, $y$ and all alternatives dominated by $y$ are no longer in $B$.
Also, no new alternative is added into $B$.
Hence the size of $B$ reduces to at most $\ell-1-\lceil (\ell-1)/2\rceil = \lfloor (\ell-1)/2\rfloor$.
Since $|B|\leq n-1$ at the beginning, the size of $B$ becomes $0$ after at most $\lceil \log_2 n\rceil$ reversals, at which point $x\in \ba(T)$.

To show the asymptotic tightness for \uc, assume that $x$ is a Condorcet loser and $T_{-x}$ is a tournament for which any dominating set has size $\Omega(\log n)$; such a tournament is known to exist \citep{Erdos63,GrahamSp71}.
Let $R\subseteq E(T)$ be a CRS for $x$ with respect to \uc.
Observe that if there is an edge $(y,z)\in R$ such that $x\not\in\{y,z\}$, then by replacing $(y,z)$ with $(y,x)$ (or simply removing $(y,z)$ if $(y,x)$ already belongs to $R$), the resulting set $R'$ is still a CRS for $x$.
Moreover, $|R'|\leq |R|$.
Therefore we may assume that all edges in $R$ are incident to $x$; let these edges be $(y_1,x),\dots,(y_{|R|},x)$.
Since $x \in \uc(T^R)$, the set $\{y_1,\dots,y_{|R|}\}$ necessarily forms a dominating set in $T_{-x}$.
It follows that $|R|\geq\Omega(\log n)$, as desired.
\end{proof}

\section{Discussion}

In this paper, we have proposed a new framework for refining tournament solutions based on the notion of margin of victory (MoV). We have determined the complexity of computing the MoV, as well as worst-case bounds on the MoV, for several common tournament solutions.
Besides the tournament solutions that we have considered, it would be interesting to study the MoV with respect to other tournament solutions such as the bipartisan set, the minimal covering set, the tournament equilibrium set, and the Markov set.

Viewing the MoV as a robustness measure, one could aim to obtain more comprehensive information about the space of all (not necessarily minimum) reversal sets. For example, one may ask \textit{how many} reversal sets of cost at most $c$ exist for a given alternative. Investigating the complexity of computing these numbers is an appealing future direction; similar counting questions have been considered in the context of knockout tournaments \citep{AzizGaMa18}.

In particular, one could use the number of minimum reversal sets as a tie-breaker for alternatives with equal MoV. Indeed, note that for some tournaments, especially small ones, the MoV in the unweighted setting may not distinguish between all winners (or non-winners).
An example is the tournament in Figure~\ref{fig:example}, where three of the four \uc winners have a MoV of $1$.
A natural way to differentiate between alternatives with the same MoV is to consider the number of minimal reversal sets for each of them---for the example above, $c$ has two minimal reversal sets ($\{(c,f)\},\{(f,d)\}$), $d$ has four ($\{(d,c)\},\{(d,b)\},\{(c,f)\},\{(b,e)\}$), and $e$ has three ($\{(c,f)\},\{(e,d)\},\{(e,c)\}$).

Understanding the counting problem is also relevant for settings in which there is uncertainty regarding the pairwise comparisons (\eg say that the direction of each edge is incorrect with some fixed probability $p<0.5$). In such a scenario, the number of reversal sets of a given size can be used to compute the winning probabilities of alternatives.

\section*{Acknowledgments}

This work was partially supported by the Deutsche Forschungsgemeinschaft under grant BR~4744/2-1 and by the European Research Council (ERC) under grant number 639945 (ACCORD). 
We would like to thank Felix Brandt and Piotr Faliszewski for helpful discussions and the anonymous reviewers for valuable feedback.

\bibliographystyle{aaai}
\bibliography{abb,mov,algo}

\clearpage
\appendix
\section{Appendix}

\subsection{Omitted Proofs from \Cref{sec:compWinners}}

\restateCopelandDestructive*

\begin{proof}
Let $x$ be the \cp winner for which we want to compute the MoV.
Consider a fixed minimal destructive reversal set $R$ with $T$ being the tournament before and $T^{R}$ the tournament after the edge reversal and let $y$ be an alternative with higher outdegree than $x$ in $T^R$. We claim that $R$ contains outgoing edges of $x$ and ingoing edges of $y$ only. Assume for contradiction that an edge which is neither outgoing of $x$ nor ingoing to $y$ is included in $R$. Then, deleting this edge from $R$ does not increase the outdegree of $x$ or decrease the outdegree of $y$ in $T^R$, a contradiction to the minimality of~$R$. 

The above observation directly implies a simple polynomial time procedure to compute a minimal destructive reversal set: Iterate over all $y \in V(T) \setminus \{x\}$ and compute the cost of a minimal reversal set that makes the outdegree of $y$ higher than $x$. Up to the choice of the edge $(x,y)$, which we handle by a case distinction, we can do so by greedily choosing outgoing edges of $x$ and ingoing edges of $y$ of lowest cost until $y$ has higher outdegree than $x$. Among all choices of $y$ we select one which induces minimum cost. 

To see the correctness of this algorithm, note that, after we fixed $y$ and decided that $(x,y) \not\in R$, reversing an edge outgoing of $x$ or ingoing to $y$ reduces the difference $|D(x)|-|D(y)|$ by $1$, and $y$ has higher outdegree than $x$ exactly when this difference becomes negative. The same argument holds for the case that $(x,y) \in R$. 
\end{proof}

We remark that the implication from right to left in Lemma \ref{lem:xcut-kkings-destr} holds only for minimum $x$-cuts and not for general $x$-cuts. In \Cref{fig:exm:xcuts} we give an example where an $x$-cut does not correspond to a \drs. In the illustrated unweighted tournament, $x$ is a $3$-king and the edge set $C=\{b,e,c\}$ is a $3$-length bounded $x$-$y$-cut and, hence, a $3$-length bounded $x$-cut. However, $C$ is not a \drs as its reversal creates a new $x$-$y$-path of length three (namely, $(a, \overline{e},d)$) and all other nodes are also still reachable in three steps. 

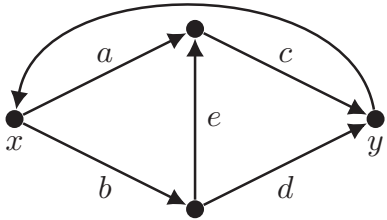
\begin{figure}[!h]
\centering
\scalebox{1.2}{
\begin{tikzpicture}

\node[circle,fill,inner sep=2pt](v1) at (0,0) {}; \node[below=2pt] at (v1){$x$}; 
\node[circle,fill,inner sep=2pt](v2) at (2,1) {}; 
\node[circle,fill,inner sep=2pt](v3) at (2,-1) {}; 
\node[circle,fill,inner sep=2pt](v4) at (4,0) {}; \node[below=2pt] at (v4){$y$}; 

\draw[->,thick] (v1) -- (v2) node [midway,above] {$a$};
\draw[->,thick] (v1) -- (v3)  node [midway,below] {$b$};
\draw[->,thick] (v2) -- (v4) node [midway,above] {$c$};
\draw[->,thick] (v3) -- (v4)  node [midway,below] {$d$};
\draw[->,thick] (v3) -- (v2)  node [midway,right] {$e$};
\draw[->,thick] (v4) to[bend right=80](v1)  node [midway,above] {};

\end{tikzpicture}}
\caption{Example showing that the implication of Lemma \ref{lem:xcut-kkings-destr} does not hold for general $k$-length bounded $x$-cuts. The edge set $\{b,e,c\}$ is a $3$-length bounded $x$-cut, but not a \drs for~$x$ with respect to $3$-kings.}\label{fig:exm:xcuts}
\end{figure}

The remainder of this section is devoted to proving \Cref{thm:destr:kkings:nphard}.
Recall that $\mathcal{P}_{x,y}(k)$ denotes the set of $x$-$y$-paths in tournament $T$ of length at most $k$. 

\restateKKingsReduction*

\begin{figure*}[t!]
\centering

\begin{tikzpicture}
\node[ultra thick,draw=red!80!black,circle,fill=white,inner sep=4pt, dash pattern = on 4pt off 1pt](v) at (0,4) {};\node[text=red!80!black,right=6pt] at (v){\Large$v$};

\node[circle,ultra thick,draw=blue,fill=white,inner sep=4pt](u) at (0,0) {};\node[blue,right=6pt] at (u){\Large$u$};

\draw[very thick] (u) -- (v);
\end{tikzpicture}\hspace{.1cm}
\vline \hspace{.1cm} \raisebox{0cm}{\scalebox{.6}{\begin{tikzpicture}

\fill [black!20, rounded corners=2ex] (-2,-1.3) rectangle (10,4);  \node[right=4pt] at (10,3){\huge$N_v$};

\node [circle, very thick, draw=black, fill=white, inner sep=1cm](d) at (0,2){}; \node[above right=1cm and 1cm] at (d){\huge$ \overline{v}_1$};
\node [circle, very thick, draw=black, fill=white, inner sep=1cm](d2) at (3.5,2){}; \node[above right=1cm and 1cm] at (d2){\huge$ \overline{v}_2$};
\node [circle, very thick, draw=black, fill=white, inner sep=1cm](d3) at (7,2){}; \node[above right=1cm and 1cm] at (d3){\huge$ \overline{v}_3$};

\node[circle,fill,inner sep=3pt](v1) at (0,-.3){}; \node[below=4pt] at (v1){\huge$ v_1$};
\node[circle,fill,inner sep=3pt](v2) at (3.5,-.3){}; \node[below=4pt] at (v2){\huge$ v_2$};
\node[circle,fill,inner sep=3pt](v3) at (7,-.3){}; \node[below=4pt] at (v3){\huge$ v_3$};
\node[circle,fill,inner sep=3pt](v4) at (9,-.3){}; \node[below=4pt] at (v4){\huge$ v_4$};

\node[circle,fill,inner sep=3pt](s) at (-3,-2.5){}; \node[below=4pt] at (s){\huge$x$};
\node[circle,fill,inner sep=3pt](t) at (11,-2.5){}; \node[below=4pt] at (t){\huge$y$};

\def \n {5}
\def \radius {1cm}
\def \margin {8} % margin in angles, depends on the radius 
\foreach \s in {1,...,\n}
{ 
      \node[circle,fill,inner sep=2pt,yshift=2cm](\s) at ({360/\n * (\s - 1)}:\radius) {}; 
}

\foreach \s in {6,...,10}
{ 
      \node[circle,fill,inner sep=2pt,xshift=3.5cm,yshift=2cm](\s) at ({360/\n * (\s - 1)}:\radius) {};
}

\foreach \s in {11,...,15}
{ 
      \node[circle,fill,inner sep=2pt,xshift=7cm,yshift=2cm](\s) at ({360/\n * (\s - 1)}:\radius) {};
}

\foreach \s/\t in {1/2,1/3,2/3,2/4,3/4,3/5,4/5,4/1,5/1,5/2,
			6/7,6/8,7/8,7/9,8/9,8/10,9/10,9/6,10/6,10/7,
			11/12,11/13,12/13,12/14,13/14,13/15,14/15,14/11,15/11,15/12}{
	\draw[<-,thick] (\s) -- (\t);}

\draw[->,very thick] (d) -- (v1);
\draw[->,very thick] (d) -- (d2);
\draw[->,very thick] (d2) -- (d3);
\draw[->,very thick] (d2) -- (v2);
\draw[->,very thick] (d3) -- (v3);
\draw[->,very thick] (d3) -- (v4);
\draw[->,very thick] (s) -- (v1);
\draw[->,very thick] (s) -- (d);
\draw[->,very thick] (v1) -- (v2);
\draw[->,very thick] (v2) -- (v3);
\draw[->,very thick] (v3) -- (v4);
\draw[->,very thick] (v4) -- (t);

\draw[->,very thick] (s) -- (v2);
\draw[->,very thick] (v3) -- (t);
 
\end{tikzpicture}}} \hspace{.1cm} \vline \hspace{.1cm} \scalebox{.6}{\begin{tikzpicture}
\tikzstyle{supernode}=[draw, ultra thick, circle,fill,inner sep=5pt,fill=white]
\tikzstyle{superarrow}=[ultra thick]

\fill [black!20, rounded corners=2ex] (1,1.4) rectangle (9,4); \node[right=4pt] at (9,3){\huge$N_v$};
\fill [black!20, rounded corners=2ex] (1,-1.4) rectangle (9,-4); \node[right=4pt] at (9,-2.4){\huge$N_u$};
%LEVEL 0 

\node[circle,fill,inner sep=3pt](s) at (-.5,0) {};\node[below=8pt] at (s){\huge$x$};

%LEVEL 1
\node[circle,fill,inner sep=3pt](v1) at (2,2) {}; \node[above right] at (v1){};%{\huge$v_1$};
\node[supernode](vb1) at (2,3.5) {}; \node[above=3pt] at (vb1){};%{\huge$\overline{v}_1$}; 
\node[circle,fill,inner sep=3pt](u1) at (2,-2) {}; \node[below right] at (u1){};%{\huge$ u_1$};
\node[supernode](ub1) at (2,-3.5) {}; \node[below=3pt] at (ub1){};%{\huge$ \overline{u}_1$};

%LEVEL 2
\node[circle,fill,inner sep=3pt](v2) at (4,2) {}; \node[above right] at (v2){};%{\huge$ v_2$};
\node[supernode](vb2) at (4,3.5) {}; \node[above=3pt] at (vb2){};%{\huge$ \overline{v}_2$}; 
\node[circle,fill,inner sep=3pt](u2) at (4,-2) {}; \node[below right] at (u2){};%{\huge$ u_2$};
\node[supernode](ub2) at (4,-3.5) {}; \node[below=3pt] at (ub2){};%{\huge$ \overline{u}_2$};

\node[circle,fill,inner sep=3pt](e1) at (4,.6) {}; \node[above right] at (e1){\huge$e_1$};
\node[circle,fill,inner sep=3pt](e2) at (4,-.6) {}; \node[below right] at (e2){\huge$e_2$};

%LEVEL 3
\node[circle,fill,inner sep=3pt](v3) at (6,2) {}; \node[above right] at (v3){};%{\huge$ v_3$};
\node[supernode](vb3) at (6,3.5) {}; \node[above=3pt] at (vb3){};%{\huge$ \overline{v}_3$}; 
\node[circle,fill,inner sep=3pt](u3) at (6,-2) {}; \node[below right] at (u3){};%{\huge$ u_3$};
\node[supernode](ub3) at (6,-3.5) {}; \node[below=3pt] at (ub3){};%{\huge$ \overline{u}_3$};

%LEVEL 4
\node[circle,fill,inner sep=3pt](v4) at (8,2) {}; \node[above right] at (v4){};%{\huge$ v_4$};
\node[circle,fill,inner sep=3pt](u4) at (8,-2) {}; \node[below right] at (u4){};%{\huge$ u_4$};

%LEVEL5
\node[circle,fill,inner sep=3pt](t) at (10.5,0) {};\node[below=4pt] at (t){\huge$y$};

%EDGES 
\draw[->,very thick] (s) -- (v1);
\draw[->,superarrow] (s) -- (vb1);

\draw[->,very thick] (s) -- (u1);
\draw[->,superarrow] (s) -- (ub1);
%\draw[->,very thick] (s) -- (u2);

\draw[->,superarrow] (vb1) -- (vb2);
\draw[->,superarrow] (vb2) -- (vb3);
\draw[->,superarrow] (vb3) -- (v4);
\draw[->,superarrow] (vb1) -- (v1);
\draw[->,superarrow] (vb2) -- (v2);
\draw[->,superarrow] (vb3) -- (v3);
\draw[->, very thick] (v1) -- (v2);
%\draw[->,very thick] (v2) -- (v3);
\draw[->,very thick] (v3) -- (v4);
%\draw[->,thick] (v3) -- (t);
\draw[->,very thick] (v4) -- (t);

\draw[->,superarrow] (ub1) -- (ub2);
\draw[->,superarrow] (ub2) -- (ub3);
\draw[->,superarrow] (ub3) -- (u4);
\draw[->,superarrow] (ub1) -- (u1);
\draw[->,superarrow] (ub2) -- (u2);
\draw[->,superarrow] (ub3) -- (u3);
\draw[->,very thick] (u1) -- (u2);
%\draw[->,thick] (u2) -- (u3);
\draw[->,very thick] (u3) -- (u4);
%\draw[->,very thick] (u3) -- (t);
\draw[->,very thick] (u4) -- (t);

\draw[->,very thick] (v2) -- (e1);
\draw[->,very thick] (u2) -- (e2);
\draw[->,very thick] (e2) to [bend left](v2);
\draw[->,very thick] (e1) to [bend right](u2);
\draw[->,very thick] (e1) -- (u3);
\draw[->,very thick] (e2) -- (v3);

\draw[->,line width=1.4mm,red!80!black,>={Latex[width=4mm,length=4mm]},dash pattern = on 4pt off 2pt] (s) -- node[below=4pt]{\huge$\boldsymbol{\ell_v}$}(v2); %lv
\draw[->,very thick,black] (v2) -- node[above=1pt]{\huge${m_v}$}(v3); %mv 
\draw[->,line width=1.4mm,red!80!black,>={Latex[width=4mm,length=4mm]},dash pattern = on 4pt off 2pt] (v3) -- node[below=4pt]{\huge$\boldsymbol{r_v}$}(t); %rv 

\draw[->,very thick,black] (s) -- node[above=2pt]{\huge${\ell_u}$}(u2); %lu 
\draw[->,line width=1.4mm,blue,>={Latex[width=4mm,length=4mm]}] (u2) -- node[below=4pt]{\huge$\boldsymbol{m_u}$}(u3); %mu 
\draw[->,very thick,black] (u3) -- node[above=4pt]{\huge${r_u}$}(t); %ru 

%LABELS

\end{tikzpicture}}

\repeatcaption{fig:kkingsReduction}{Illustration of the construction used in the proof of \Cref{thm:destr:kkings:nphard} for the case $k=4$. For any graph $G$ (left image), a tournament $T$ is constructed by introducing node gadgets and edge gadgets as follows. A \emph{node gadget} $N_v$ consists of four nodes $v_1, v_2, v_3, v_4$ and three supernodes $\overline{v}_1, \overline{v}_2, \overline{v_3}$, where the latter are tournaments themselves. The center image shows the node gadget for node $v$. An \emph{edge gadget} for $e=\{u,v\}$ consists of two nodes $e_1, e_2$ and edges connecting the node gadgets of $u$ and $v$; see the right image. Nodes $x$ and $y$ are connected to all node gadgets as illustrated. All omitted edges point ``backwards'' (from right to left) and the direction of vertical edges, if not specified, can be chosen arbitrarily. 
% We ask for the MoV of $x$ with respect to $4$-kings.
}
%\label{fig:kkingsReduction}
\end{figure*}

% \begin{proof}
Just as \citet{BaierEr10}, we reduce from the vertex cover problem. Given an undirected graph $G=(V,E)$, a subset of the nodes $U \subseteq V$ is called a \emph{vertex cover} if for each edge in $E$, at least one of its endnodes is contained in $U$. The problem is to determine the minimum cardinality of a vertex cover. %\citetAppendix{BaierEr10}

The proof is divided into three parts. We start by showing the reduction for $k=4$, then extend the construction to arbitrary constants $k$, and finally argue that the we can still carry out the reduction even when we restrict ourselves to cases in which $k \geq n^{1-\epsilon}$ for arbitrarily small $\epsilon>0$. 

\paragraph{Part I.}
% \textit{Part I.}\hspace{1ex}
Let $k=4$. From a given instance $G$ of the vertex cover problem, we construct an instance $(T,x)$ of the MoV for $4$-kings problem as follows. (An illustration of the construction can be found in the three images of \Cref{fig:kkingsReduction}.) 
For ease of presentation we define $n_G := |V(G)|$. For every node $v \in V(G)$ we introduce a \emph{node gadget}, indicated by a grey box in \Cref{fig:kkingsReduction} and consisting of four nodes $v_1,v_2,v_3$ and $v_4$ as well as three \emph{supernodes}, $\overline{v}_1, \overline{v}_2$ and $\overline{v}_3$. A supernode is itself a tournament consisting of $2 n_G + 1$ nodes which are arranged within a circle such that each node has outgoing edges towards the next $n_{G}$ nodes on the circle and ingoing edges from all other nodes. See the center image in \Cref{fig:kkingsReduction} for a close-up of a node gadget with $n_G=2$. Moreover, we introduce two nodes $x$ and $y$. Corresponding to the node gadget there exist edges 
\[(x,v_1),(x,v_2),(v_1,v_2),(v_2,v_3),(v_3,v_4),(v_4,y),(v_3,y)\] 
as well as \emph{superedges} (\ie edges that connect at least one supernode)
\[(x,\overline{v}_1),(\overline{v}_1,\overline{v}_2),(\overline{v}_2,\overline{v}_3),(\overline{v}_1,v_1),(\overline{v}_2,v_2),(\overline{v}_3,v_3),(\overline{v}_3,v_4).\] 
More precisely, a superedge $(\overline{u},\overline{v})$ is a set of edges going from all nodes in $\overline{u}$ to all nodes in $\overline{v}$. We extend the graph to a tournament by letting all unspecified edges point ``backwards'' (from right to left). To this end see the horizontal arrangement within \Cref{fig:kkingsReduction}. In case that two nodes are on the same vertical line and the edge between them has not been specified previously, it can be chosen arbitrarily. 

For every edge $e=(u,v) \in E$ we introduce an \emph{edge gadget}, consisting of two nodes $e_1$ and $e_2$ and the (super)edges $(v_2,e_1),(e_1,u_2),(e_1,u_3),(\overline{v}_2,e_1)$ and analogously $(u_2,e_2),(e_2,v_2),(e_2,v_3),(\overline{u}_2,e_2)$. For the sake of clarity we omit the superedges $(\overline{v}_2,e_1)$ and $(\overline{u}_2,e_2)$ in the illustration of \Cref{fig:kkingsReduction}. For all nodes $w \in V(G) \setminus \{u,v\}$ we add edges $(e_1,w_2)$ and $(e_2,w_2)$. Again, all unspecified edges points backwards and the direction of non-specified vertical edges can be chosen arbitrarily. 

Lemma \ref{lem:xcut-kkings-destr} implies that the cardinality of a minimum $4$-bounded $x$-cut is equal to the cardinality of a minimum \drs for $x$ with respect to $4$-kings. Hence, showing the following claim suffices to prove the theorem for $k=4$. 

\begin{claim}
For $c \leq n_G$ there exists a vertex cover of size $c$ in $G$ iff there exists a $4$-bounded $x$-cut of size $c+n_G$ in $T$.
\end{claim}

\begin{proof}[Proof of Claim]
We start by showing the implication from left to right. Let $U \subseteq V(G)$ be a vertex cover in $G$. For every node $v \in V(G)$ we name the following three edges inside the node gadget $\ell_v := (x,v_2), m_v := (v_2,v_3)$ and $r_v:=(v_3,y)$, these are depicted by red and dashed edges in \Cref{fig:kkingsReduction}. 
We construct the edge set $C \subseteq E(T)$, which we will show is a $4$-bounded $x$-$y$-cut, by iterating over all $v \in V(G)$: If $v \in U$, we choose $\ell_v$ and $r_v$ to be in the set $C$. On the other hand, if $v \not\in U$, we include $m_v$ in the set $C$. It is easy to see that $|C| = n_{G} + c$ and it remains to show that $C$ is a $4$-bounded $x$-$y$-cut.

Since the tournament $T$ contains a very large number of $x$-$y$-paths, we argue in the following that, when caring about $\mathcal{P}_{x,y}(4)$ only, we can virtually restrict ourselves to the set of ``visible'' paths in \Cref{fig:kkingsReduction}, i.e., paths that do not contain omitted edges. See \Cref{tab:distx} and \Cref{tab:disty} which contain distances from $x$ and to $y$ within the visible subgraphs of $T$, respectively. The columns of the tables correspond to the positions of the nodes within the horizontal alignment in  \Cref{fig:kkingsReduction}. 

\newcolumntype{?}{!{\vrule width 1pt}}
\begin{table}[!h] \centering
  \begin{tabular}{? c ? c ? c ? c ?}
   \specialrule{.1em}{.05em}{.05em} 
    Pos. 1 & Pos. 2 & Pos. 3 & Pos. 4 \\ \specialrule{.1em}{.05em}{.05em} 
   $\overline{v}_1  \;| \; 1 $ & $\overline{v}_2  \;| \; 2$ & $\overline{v}_3  \;| \; 3$ &  \\  \hline %
       $v_1  \;| \; 1$ & $v_2  \;| \; 1 $ & $v_3  \;| \; 2$ & $v_4  \;| \; 3$ \\ \hline %
    & $e_1  \;| \; 2 $ && \\ \specialrule{.1em}{.05em}{.05em} 
  \end{tabular}
  \caption{Distances from $x$ to nodes in node gadget $N_v$.} \label{tab:distx}
  \end{table}
  
  \begin{table}[!h]\centering
    \begin{tabular}{? c ? c ? c ? c ?}
   \specialrule{.1em}{.05em}{.05em} 
    Pos. 1 & Pos. 2 & Pos. 3 & Pos. 4 \\ \specialrule{.1em}{.05em}{.05em} 
   $\overline{v}_1  \;| \; 4 $ & $\overline{v}_2  \;| \; 3$ & $\overline{v}_3  \;| \; 2$ &  \\ \hline% \specialrule{.1em}{.05em}{.05em} 
       $v_1  \;| \; 3$ & $v_2  \;| \; 2 $ & $v_3  \;| \; 1$ & $v_4  \;| \; 1$ \\ \hline %\specialrule{.1em}{.05em}{.05em} 
    & $e_1  \;| \; 2 $ && \\ \specialrule{.1em}{.05em}{.05em}  
  \end{tabular}
  \caption{Distances from nodes in node gadget $N_v$ to $y$.} \label{tab:disty}
\end{table}

We first claim that the introduction of backward and vertical edges does not change these distances. To see this, let $v_i$ and $v_j$ be (super)nodes at position $i$ and $j$ respectively, where $i \leq j$. For any choice of $v_i$ and $v_j$ it holds that $dist(x,v_j) + 1 \geq dist(x,v_i)$ and $dist(v_i,y) + 1 \geq dist(v_j,y)$, where $dist()$ denotes the distance between two nodes. Hence the backward and vertical edges cannot decrease these distances.

Second, we claim that no backward edge is included in a path in $\mathcal{P}_{x,y}(4)$. Assume for contradiction that there exists such a path with backward edge $e$. First, assume that the head of $e$ is at position $1$. Observe that $e$ cannot be the first edge in the path. However, the minimum distance from position $1$ to $y$ is $3$, a contradiction. Second, assume that the head of $e$ is at position $2$ and observe that $e$ needs to be the third or fourth edge in the path since the path needs to reach position $3$ or $4$ before using edge $e$. However, the minimum distance from position $2$ to $y$ is $2$, a contradiction. Lastly, assume that the head of edge $e$ is in position $3$ and observe that edge $e$ needs to already be the fourth edge in the path since position $4$ cannot be reached in less than $3$ steps, a contradiction. 

Third, we claim that there exist exactly two types of vertical omitted edges which are included in paths in $\mathcal{P}_{x,y}(4)$. A vertical edge $(u,v)$ is included in a path in $\mathcal{P}_{x,y}(4)$ iff $dist(x,u) + dist(v,y) \leq 3$. Looking at \Cref{tab:distx} and \ref{tab:disty} we see that the only candidates for such an edge are $\{v_2,u_2\}$, $\{v_2,e_1\}$, $\{v_2,e_2\}$ and $\{v_3,u_3\}$. Since the edges $(v_2,e_1)$ and $(e_2,v_2)$ have already been specified, there remain only $\{v_2,u_2\}$ and $\{v_3,u_3\}$. No matter how the directions of these edges are chosen, they create two new paths of length four, e.g., $\{\ell_v,(v_2,u_2),m_u,r_u\}$ and $\{\ell_u,m_u,(u_3,v_3),r_v\}$. 

Lastly, we claim that no node contained in a supernode can be used in any path of $\mathcal{P}_{x,y}(4)$. To see this, note that the distances to $x$ and from $y$ sum up to five for all supernodes.

Having this in mind, we show that every path in $\mathcal{P}_{x,y}(4)$ includes some edge in $C$. First consider all paths in $\mathcal{P}_{x,y}(4)$ which only use edges within one node gadget, say the one corresponding to $v \in V$. These are exactly the paths $\{\ell_v,m_v,r_v\}$, $\{\ell_v,m_v,(v_3,v_4),(v_4,y)\}$ and $\{(x,v_1),(v_1,v_2),m_v,r_v\}$. All of these contain at least one of the edges $\ell_v,m_v$ and $r_v$ and hence, independent of whether $v \in U$ or not, the paths include an edge in $C$. Second, consider paths in $\mathcal{P}_{x,y}(4)$ which use edges within two node gadgets $u$ and $v$ which are not neighboring in the graph $G$. How these paths look exactly depends on the direction of the edge betweeen $u_2$ and $v_2$ as well as the edge between $u_3$ and $v_3$. W.l.o.g. we assume that they are $\{\ell_v,(v_2,u_2),m_u,r_u\}$ and $\{\ell_u,m_u,(u_3,v_3),r_v\}$. They have the property that for one of the nodes, say $u$, they contain either both $\ell_u$ and $m_u$, or both $m_u$ and $r_u$. Hence, independent of whether $u \in U$ or not, the paths include an edge in $C$. Lastly, consider paths in $\mathcal{P}_{x,y}(4)$ which use two node gadgets corresponding to neighboring nodes in $G$, say $u$ and $v$. These paths are $\{\ell_v,(v_2,e_1),(e_1,u_3),r_u\}$ and $\{\ell_u,(u_2,e_2),(e_2,v_3),r_v\}$ as well as (with the same w.l.o.g. assumption as earlier in this paragraph), $\{\ell_v,(v_2,u_2),m_u,r_u\}$ and $\{\ell_u,m_u,(u_3,v_3),r_v\}$. All of them have the property that they contain either both $\ell_v$ and $r_u$, or both $\ell_u$ and $r_v$. Since $U$ is a vertex cover we know that at least one of the pairs $\ell_v,r_v$ and $\ell_u,r_u$ is included in $C$ and therefore the paths contain an edge in $C$. We summarize that $C$ is a $4$-bounded $x$-$y$-cut of size $|U| + n_G$. 

We turn to prove the direction of the Claim from right to left. Let $C \subseteq V$ be a $4$-bounded $x$-cut of size $n_{G} + c$ with $c \leq n_{G}$. For any node $z \in V(T)\setminus \{y\}$, the set $C$ cannot be a $4$-bounded $x$-$z$-cut because there exist at least $2n_{G} +1 > n_{G}+c$ disjoint $4$-bounded $x$-$z$-paths due to the introduction of the supernodes. Hence, $C$ is a $4$-bounded $x$-$y$-cut. In the following, we transform $C$ so that it only contains edges of type $\ell_v$, $m_v$ and $r_v$. 

In our first step, we ensure that no edges connecting node gadgets or having an endnode from an edge gadget are included in $C$. Assume that $C$ contains an edge $e$ connecting gadgets corresponding to the nodes $u$ and $v$, where we consider $e_2$ to correspond to node $u$ and $e_1$ to node $v$. If $e$ is not included in any path of $\mathcal{P}_{x,y}(4)$, we simply delete $e$ from $C$. Otherwise, $e$ is contained in exactly one path from $\mathcal{P}_{x,y}(4)$. (To see this, consult the complete characterization of $\mathcal{P}_{x,y}(4)$ in the previous part of the proof.) Moreover, this path contains exactly one of $\ell_u$ and $\ell_v$. We replace $e$ by the edge $\ell_u$ or $\ell_v$, respectively, and obtain another $4$-bounded $x$-$y$-cut of the same size. 

In our second step, we guarantee that within the node gadget for each node $v$, either edge $m_v$ or edges $\ell_v$ and $r_v$ are selected. If one edge from this gadget is selected, it needs to be the edge $m_v$ since otherwise there exists at least one path in $\mathcal{P}_{x,y}(4)$ which does not contain an edge in $C$. If two or more edges are selected, we instead select edges $\ell_v$ and $r_v$, since all paths that contain an edge from the node gadget either contain $\ell_v$ or $r_v$. Hence, we obtain a new $4$-bounded $x$-$y$-cut of size at most the size of the previous cut.

After the transformation of $C$, we derive a vertex cover $U$ in the graph $G$ of size $c$. For each node $v \in V(G)$ we include $v$ in $U$ iff $\ell_v$ and $r_v$ are included in $C$. Clearly $|U| = |C| - n_{G}$. In case we previously reduced the cardinality of $C$, we simply add nodes to $U$ until $|U| = c$. Now, assume for contradiction that $U$ is not a vertex cover, i.e., there exists an edge $\{u,v\} \in E(G)$ such that $u,v\not\in U$. We conclude that $C$ contains both $m_u$ and $m_v$ and no other edge from the node gadgets of $u$ and $v$. Then the path $\{\ell_v,(v_2,e_1),(e_1,u_3),r_u\}$ does not contain an edge in $C$, a contradiction to $C$ being a $4$-bounded $x$-$y$-cut. 
\end{proof}

\begin{figure*}[htb]
\centering
\scalebox{1}{%\begin{minipage}{.6\textwidth}
\begin{tikzpicture}

%\tikzstyle{supernode}=[draw, thick, regular polygon,regular polygon sides=8,fill,inner sep=3pt,fill=white]
\tikzstyle{supernode}=[draw, thick, circle,fill,inner sep=3.5pt,fill=white]
\tikzstyle{superarrow}=[thick]

\fill [black!10, rounded corners=2ex] (-7,0) rectangle (7,4);
\fill [black!20, rounded corners=2ex] (-1,0) rectangle (7,4);

\node[circle,fill,inner sep=2pt](s) at (-8,1) {};\node[below=2pt] at (s){$x$};
\node[circle,fill,inner sep=2pt](t) at (8,1) {};\node[below=2pt] at (t){$y$};

\node[supernode](svb1) at (-6,3) {}; \node[above right] at (svb1){${\overline{v}_1}^{(3)}$};
\node[circle,fill,inner sep=2pt](sv11) at (-6,2) {}; \node[above right] at (sv11){$ {v_1}^{(3)}$};
\node[circle, fill, inner sep = 2pt](sv21) at (-6,1) {}; \node[above right] at (sv21){$ {v_2}^{(3)}$}; 

\node[supernode](svb2) at (-4,3) {}; \node[above right] at (svb2){${\overline{v}_1}^{(2)}$};
\node[circle,fill,inner sep=2pt](sv12) at (-4,2) {}; \node[above right] at (sv12){${v_1}^{(2)}$};
\node[circle, fill, inner sep = 2pt](sv22) at (-4,1) {}; \node[above right] at (sv22){${v_2}^{(2)}$}; 

\node[supernode](svb3) at (-2,3) {}; \node[above right] at (svb3){${\overline{v}_1}^{(1)}$};
\node[circle,fill,inner sep=2pt](sv13) at (-2,2) {}; \node[above right] at (sv13){${v_1}^{(1)}$};
\node[circle, fill, inner sep = 2pt](sv23) at (-2,1) {}; \node[above] at (sv23){${v_2}^{(1)}$}; 

\node[circle,fill,inner sep=2pt](v1) at (0,2) {}; \node[above right] at (v1){$v_1$};
\node[supernode](vb1) at (0,3) {}; \node[above right] at (vb1){$\overline{v}_1$}; 

\node[supernode](vb2) at (2,3) {}; \node[above=3pt] at (vb2){$\overline{v}_2$}; 
\node[circle,fill,inner sep=2pt](v2) at (2,2) {}; \node[above right] at (v2){$v_2$};

\node[supernode](vb3) at (4,3) {}; \node[above right] at (vb3){$\overline{v}_3$}; 
\node[circle,fill,inner sep=2pt](v3) at (4,2) {}; \node[above right] at (v3){$v_3$};
\node[circle,fill,inner sep=2pt](v4) at (6,2) {}; \node[above right] at (v4){$v_4$};

%EDGES 
\draw[->,thick] (s) -- (sv11);
\draw[->,superarrow] (s) -- (svb1);
\draw[->,thick] (s) -- (sv21);

\draw[->,thick] (sv11) -- (sv12);
\draw[->,superarrow] (svb1) -- (svb2);
\draw[->,thick] (sv21) -- (sv22);

\draw[->,thick] (sv12) -- (sv13);
\draw[->,superarrow] (svb2) -- (svb3);
\draw[->,thick] (sv22) -- (sv23);

\draw[->,thick] (sv13) -- (v1);
\draw[->,superarrow] (svb3) -- (vb1);
\draw[->,line width=.8mm,red!80!black,>={Latex[width=3mm,length=3mm]},dash pattern = on 4pt off 1pt] (sv23) -- node[below]{$\boldsymbol{\ell_v}$}(v2); %lv

\draw[->,line width=.8mm,blue!80!black,>={Latex[width=3mm,length=3mm]}] (v2) -- node[below]{$\boldsymbol{m_v}$}(v3); %mv

\draw[->,line width=.8mm,red!80!black,>={Latex[width=3mm,length=3mm]},dash pattern = on 4pt off 1pt] (v3) -- node[below]{$\boldsymbol{r_v}$}(t); %rv

\draw[->,thick] (v1) -- (v2);
\draw[->,thick] (vb1) -- (v1);
\draw[->,thick] (vb2) -- (v2);
\draw[->,thick] (vb3) -- (v3);

\draw[->,thick] (vb1) -- (vb2);
\draw[->,thick] (vb2) -- (vb3);
\draw[->,thick] (vb3) -- (v4);

\draw[->,thick] (v3) -- (v4);
\draw[->,thick] (v4) -- (t);
\draw[->,thick] (sv23) -- (v1);
%LABELS

\end{tikzpicture}
%\end{minipage}
}
\caption{Example illustration of the extended node gadget for $k=7$ as introduced in the proof of \Cref{thm:destr:kkings:nphard}.}\label{fig:destr:kkings:reduction2}
\end{figure*}
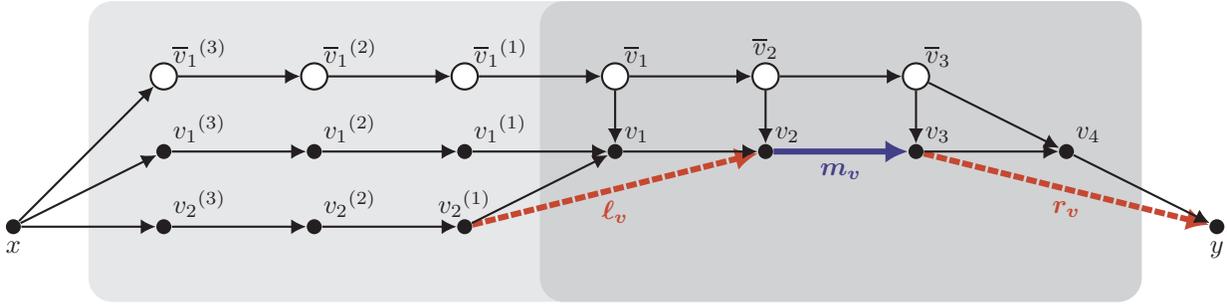

\paragraph{Part II.}
% \textit{Part II.}\hspace{1ex}
We turn to show how we can adjust the construction for fixed $k > 4$. We change the definition of a node gadget by extending it with a set of preceding (super)nodes \[\{{\overline{v}_1}^{(i)},{v_1}^{(i)},{v_2}^{(i)} \; | \; \forall \; i \in \{1,\dots, k-4\} \text{ and } \; \forall \; v \in V\}.\] 

(Super)edges go from $x$ to all nodes with superscript $(k-4)$ and more generally from a node with superscript $(i+1)$ to the node of the same type with superscript $(i)$ if they correspond to the same node $v \in V(G)$. Moreover, a superedge points from ${\overline{v}_1}^{(1)}$ towards $\overline{v}_1$, an edge from $v_1^{(1)}$ towards $v_1$ and edges from $v_2^{(1)}$ towards $v_1$ and $v_2$. For the connections among the nodes $v_1, v_2, v_3, v_4, \overline{v}_1, \overline{v}_2, \overline{v}_3,y$ we use the same edges as in the case $k=4$. Edge gadgets are defined exactly as in the case $k=4$. All non-specified edges point backwards and vertical edges can be chosen arbitrarily. For every $v \in V(G)$ we define $\ell_v := (v_2^{(1)},v_2)$, $m_v := (v_2,v_3)$ and $r_v:= (v_3, y)$. See \Cref{fig:destr:kkings:reduction2} for an illustration of the extended node gadget.  We show the same claim as previously and omit analogous arguments. 

\begin{claim}
For $c \leq n_G$ there exists a vertex cover of size $c$ in $G$ iff there exists a $k$-bounded $x$-cut of size $c+n_G$ in $T$.
\end{claim}

\begin{proof}[Proof of Claim]
We start by showing the implication from left to right. Let $U \subseteq V(G)$ be a vertex cover. Analogously to before, we construct $C \subseteq E(T)$ by choosing for every $v \in V(G)$ the edges $\ell_v$ and $r_v$ whenever $v \in U$ and the edge $m_v$ when $v \not \in U$. While it is easy to see that $|C| = |U| + n_G$, we need to show that $C$ is a $k$-bounded $x$-$y$-cut. Hence, we are interested in the set $\mathcal{P}_{x,y}(k)$ and need to show that $C$ intersects each of its paths at least once. We aim to compare the set $\mathcal{P}_{x,y}(k)$ to the set of $4$-bounded $x$-$y$-paths which we characterized in the proof for the claim for $k=4$. To this end let $T_4$ be the tournament which we constructed in the first part of the proof. First, consider only the ``visible" subgraph of $T$ and note that, in comparison to the graph $T_4$, the nodes $v_1,v_2,v_3,v_4,e_1,\overline{v}_1,\overline{v}_2$ and $\overline{v}_3$ have a distance from $x$ which is increased by exactly $k-4$ while the distance to $y$ is the same as previously. Inserting backwards and vertical edges does not change these distances, due to the same arguments as before. Since the rest of the structure of the node gadget is equivalent, the subpaths of the paths in $\mathcal{P}_{x,y}(k)$ that are within the original node gadget (depicted by a darker grey box in \Cref{fig:destr:kkings:reduction2}) correspond to the $4$-bounded $x$-$y$-paths in $T_4$. Due to the same arguments as in the first part, $C$ is a $k$-bounded $x$-$y$-cut. 

We turn to prove the implication from right to left. Let $C$ be an $x$-cut of size $c + n_G$ with $c \leq n_G$; its existence is guaranteed by the previously shown implication from left to right and the fact that there always exists a vertex cover of size $n_G$. Analogously to the case $k=4$, we show that we can modify $C$ so that it only contains edges of type $\ell_v$, $m_v$ and $r_v$ and is still a $x$-cut with no greater cost. First, note that for every $z \in V(T) \setminus \{y\}$ there exist at least $2 n_G + 1$ disjoint $k$-bounded $x$-$z$-paths and hence $C$ is in particular an $x$-$y$-cut. 

In our first step, we ensure that no edges connecting node gadgets or having an endnode from an edge gadget are included in $C$. We follow a very similar argument as in the previous proof. First, assume that $C$ contains an edge $e$ connecting gadgets corresponding to the nodes $u$ and $v$, where we consider $e_2$ to correspond to node $u$ and $e_1$ to node $v$. If $e$ is a backwards edge it is not contained in any path in $\mathcal{P}_{x,y}(k)$ and we simply delete it from $C$. If $e$ is a vertical edge, then $e$ is included in at most one path in $\mathcal{P}_{x,y}(k)$. This is due to the fact that when moving from the gadget of $u$ to the gadget of $v$, the subpath from $x$ to the first node that is reached in node gadget $v$ is necessarily one step longer than the shortest possible path. Hence, the rest of the path is completely determined, as it needs to choose the unique shortest path from $x$ to $e$ as well as the unique shortest path from $e$ to $y$ in order to fulfill the length bound. Moreover, this unique path contains either edge $\ell_u$ or edge $\ell_v$, and we replace edge $e$ by $\ell_u$ or $\ell_v$, respectively. 

The second step, in which we guarantee that for every node gadget either the pair $\ell_v, r_v$ or the edge $m_v$ is chosen, proceeds completely analogously to the case $k=4$. We complete the proof by the same argument as before, showing that we can translate $C$ to a vertex cover $U \subseteq E(G)$ of size $|U| = c$. 
 \end{proof}

\paragraph{Part III.}

It remains to argue that even if we restrict ourselves to the problem with $k \geq n^{(1-\epsilon)}$, for any fixed $\epsilon>0$, we can still carry out the previously explained reduction in polynomial time. Let $\epsilon >0$ be given and the size of the vertex cover instance $G$ be denoted by $n_G := |V(G)|$ and $m_G := |E(G)|$. Moreover, we define $n_k$ to be the number of nodes of the tournament which we construct for a given $k$. We obtain 
\begin{align*}
    n_k & = (k-1)(2n_G + 1)n_G + (2k-4)n_G + 2m_G + 2\\
    &= (2 n_G^2 + 3n_G) k + (- 2n_G^2 + - 5n_G + 2m_G +2)\\
    &\leq \alpha k,
\end{align*}
where $\alpha := (2n^2_G + 3m_G)$.  Choose the smallest $k \in \mathbb{N}_{\geq 4}$ such that \[k \geq \alpha^{(1-\epsilon)/\epsilon}.\] This is still polynomial in $n_G$ and $m_G$ while it implies that \[k = (k^{\epsilon/(1-\epsilon)}k)^{(1-\epsilon)} \geq (\alpha k)^{(1-\epsilon)} \geq n_k^{(1-\epsilon)}.\] This concludes the proof of \Cref{thm:destr:kkings:nphard}. 

\paragraph{Comparison to the Reduction of \citet{BaierEr10}} %\citetAppendix{BaierEr10}
The previous proof is strongly based on a proof presented by \citet{BaierEr10} showing that computing the size of minimum $\ell$-length bounded $s$-$t$-cuts is NP-hard to approximate for $\ell \in \{4, \dots, \lfloor n^{(1-\epsilon)} \rfloor\}$, where $n$ is the size of the length bounded cut instance constructed in the reduction and $\epsilon>0$ can be arbitrarily small. In the following we discuss our main adjustments. % \citeAppendix{BaierEr10}

First, note that the minimum $\ell$-length bounded $s$-$t$-cut problem does not require the graph to be a tournament, which is why we needed to alter the construction by introducing backwards and vertical edges. This increased the number of paths significantly. 

Second, the problem discussed by \citet{BaierEr10} specifies two nodes $s$ and $t$ which ought to be separated, while our problem specifies only one node $x$ which should be separated from some node in $V(T) \setminus \{x\}$. To this end we introduced supernodes to help guarantee that all $x$-$z$-cuts for $z \in V(T)\setminus \{y,x\}$ are significantly more expensive than a minimum $x$-$y$-cut. In this way we artificially fix the node $y$ to be separated from $x$. The introduction of supernodes in turn lead to a different extension of the node gadget for $k>4$. Beyond that, the construction is very similar to the one by \citet{BaierEr10}. % \citetAppendix{BaierEr10} \citetAppendix{BaierEr10}

\subsection{Omitted Proofs from \Cref{sec:compNonWinners}}

\restateCopelandConstructive*

\begin{figure*}[t!]
\centering
\scalebox{.98}{
\begin{tikzpicture}
%SOURCE 

\node[circle,fill,inner sep=2pt](s) at (-3,7.8) {};\node[below=2pt] at (s){$s$}; \node[above=15pt] at (s){$-n(n-1)/2$};

%LEFTSIDE
\node[circle,fill,inner sep=2pt](e1) at (0,10) {}; \node[below=4pt] at (e1){$(u,v)$}; \node[above=2pt] at (e1){$0$};
\node[circle,fill,inner sep=2pt](e2) at (0,8.7) {}; \node[below=4pt] at (e2){$(u,y)$}; \node[above=2pt] at (e2){$0$};
\node[circle,fill,inner sep=2pt](e3) at (0,7.4) {}; \node[above] at (e3){};
\node[circle,fill,inner sep=2pt](e5) at (0,6.1) {}; \node[above] at (e5){};

%dots 
\node[circle,fill,inner sep=.5pt](do1) at (0,5.6) {};
\node[circle,fill,inner sep=.5pt](do2) at (0,5.5) {};
\node[circle,fill,inner sep=.5pt](do3) at (0,5.4) {};

\node[circle,fill,inner sep=2pt](e4) at (0,4.6) {}; \node[above] at (e4){};

%RIGHTSIDE
\node[circle,fill,inner sep=2pt](u) at (5,9.5) {}; \node[below=2pt] at (u){$u$}; \node[above=2pt] at (u){$0$};
\node[circle,fill,inner sep=2pt](v) at (5,8.5) {}; \node[below=2pt] at (v){$v$}; \node[above=2pt] at (v){$0$};
\node[circle,fill,inner sep=2pt](y) at (5,7.5) {}; \node[below=2pt] at (y){$y$}; \node[above=2pt] at (y){$0$};

%dots 
\node[circle,fill,inner sep=.5pt](d1) at (5,7) {};
\node[circle,fill,inner sep=.5pt](d2) at (5,6.9) {};
\node[circle,fill,inner sep=.5pt](d3) at (5,6.8) {};

\node[circle,fill,inner sep=2pt](v4) at (5,6.2) {};\node[above] at (v4){};
\node[circle,fill,inner sep=2pt](x) at (5,5) {}; \node[below=2pt] at (x){$x$}; \node[above=3pt] at (x){$c$};

%SINK
\node[circle,fill,inner sep=2pt](t) at (8,7.8) {};\node[below=2pt] at (t){$t$}; \node[above=14pt] at (t){$n(n-1)/2- c$};

%EDGES 
\draw[->,thick] (e1) -- (u) node [midway,above] {$0$};
\draw[->,thick] (e1) -- (v)  node [pos=.85,above] {$w{(u,v)}$};
\draw[->,thick] (e2) -- (u) node [near start,above] {$0$};
\draw[->,thick] (e2) -- (y)  node [pos=.85,above] {$w{(u,y)}$};
%SMALL EDGES 

%LAST EDGES 
\draw[->,thick] (u) -- (t);
\draw[->,thick] (v) -- (t);
\draw[->,thick] (y) -- (t);
\draw[->,thick] (v4) -- (t);

\draw[<-,thick] (e1) -- (s);
\draw[<-,thick] (e2) -- (s);
\draw[<-,thick] (e3) -- (s);
\draw[<-,thick] (e4) -- (s);
\draw[<-,thick] (e5) -- (s);

%LABELS
\node at (0,11.2) {$E(T)$};
\node at (5,11.2) {$V(T)$};

\end{tikzpicture}
}
\caption{An illustration of the construction in Theorem~\ref{thm:constructive-CL}. Each node has two labels, i.e., its name below and its balance above the corresponding circle. The cost of an edge is notated above its associated arrow and capacities of edges are omitted.}
\label{fig:constructive-CL}
\end{figure*}
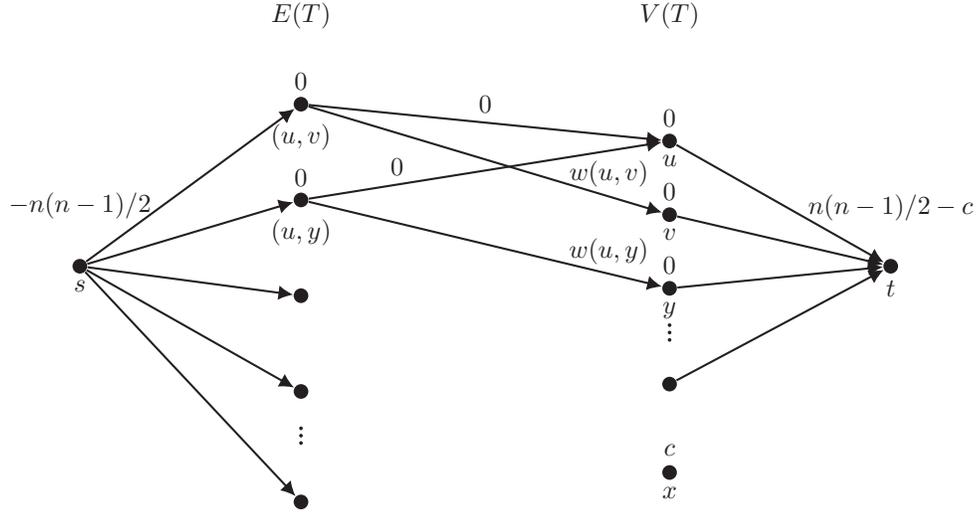

\begin{proof}
We aim to compute the MoV of the Copeland non-winner $x$. Any member of the Copeland set needs to have an outdegree of at least $\lceil (n-1)/2 \rceil$. We iterate over all $c \in \{\lceil (n-1)/2 \rceil, \dots, n\}$ and compute the minimum cost of making $x$ a Copeland winner given that the outdegree of all Copeland winners is $c$ after the reversal. To this end we construct a network $G$ where $V(G) = E(T) \cup V(T) \cup \{s,t\}$ with $E(T)$ being the edge set of the tournament $T$. 

We consider a slightly non-standard definition of a network in which edges have associated costs as well as capacities and nodes have balances. This definition allows to search for $b$-flows of minimum cost and is discussed in detail, e.g., by \citet{KorteVy06}. Informally speaking, the balance of a node corresponds to the amount of flow this node absorbs (or rather produces in case of a negative value) and a $b$-flow is a flow which respects the induced constraints. We define the balances in our network construction as follows: %\citetAppendix{KorteVy06}
\[b_v = \begin{cases} -n(n-1)/2 & \text{if } v = s\\ c & \text{if } v=x \\ n(n-1)/2 - c &\text{if } v=t \\0 &\text{else} .\end{cases}\]
There exists an edge from $e \in E(T)$ to $v \in V(T)$ iff $v$ is one of the endnodes of the edge $e$ in the tournament graph $T$. The edge $(e,v)$ has cost $0$ if $v$ is the tail of edge $e$ and otherwise it is equal to the cost of reversing $e$ in the tournament graph. All of these edges have capacity $1$. In addition there exist edges from $s$ to each node in $E(T)$ with zero cost and capacity $1$, as well as edges from each node in $V(T)\setminus \{x\}$ to $t$ with zero cost and capacity $c$.

We claim that from an integral $b$-flow in $G$ with capacity threshold $c$, we can construct a constructive reversal set for $x$ of equal cost such that $x$ is a Copeland winner with outdegree $c$ in $T^R$ and vice versa. Consider the illustration of the construction in \Cref{fig:constructive-CL}, in which the nodes are arranged in four levels, where the first level contains the source $s$, the second-level nodes correspond to the edges in $T$, the third-level nodes correspond to the nodes in $T$ and the last level contains the sink $t$. The nodes in the second layer have two outgoing edges in the network, representing the choices between \emph{keeping the direction of the corresponding edge in $T$ as it currently is} and \emph{reversing the edge} in the tournament. Any feasible integral flow can only send flow along one of the edges. 

More precisely, given an integral $b$-flow, the reversal set $R$ is determined by the edges pointing from the second to the third layer with non-zero cost and flow value $1$. The amount of flow that reaches a node in the third level corresponds exactly to the outdegree of this node in the tournament $T^R$. Since edges from the third to the fourth level have capacity $c$, a feasible flow guarantees that any node in $V(T)$ has outdegree at most $c$. Moreover, the node $x$ has a balance of $c$ and therefore it has outdegree exactly $c$ in $T^R$. Hence, $R$ is a constructive reversal set with cost equal to the cost of the flow. The other direction follows due to similar arguments. 

Integral $b$-flows of minimum cost can be found in polynomial time, for example by the minimum mean cycle-cancelling algorithm \cite{Klein67,GoldbergTa89} \footnote{The minimum mean cycle-cancelling algorithm computes integral $b$-flows if $b$ and the capacities in the network are integers.}. After repeating the construction for all $c \in \{\lceil (n-1) \rceil / 2, \dots, n\}$, we choose a constructive reversal set with minimum cost. 
\end{proof}

\end{document}